\documentclass[11pt, draftcls, onecolumn]{IEEEtran}
\pdfoutput=1
\usepackage{amsmath,epsfig}
\usepackage{cite}
\usepackage{graphicx,subfigure}
\usepackage{amssymb}
\usepackage{multirow}

\begin{document}

\newtheorem{lem}{Lemma}
\newtheorem{prop}{Proposition}
\newtheorem{cor}{Corollary}
\newtheorem{remark}{Remark}
\newtheorem{defin}{Definition}
\newtheorem{thm}{Theorem}

\newcounter{MYtempeqncnt}

\title{Joint Wireless Information and Energy Transfer in a $K$-User MIMO Interference Channel}

\author{Jaehyun Park,~\IEEEmembership{Member,~IEEE,}
     Bruno Clerckx,~\IEEEmembership{Member,~IEEE,}
\thanks{J. Park and B. Clerckx are with the Department of Electrical and Electronic
Engineering, Imperial College London, South Kensington Campus
London SW7 2AZ, United Kingdom (e-mail:\{j.park, b.clerckx\}@imperial.ac.uk)}
}

%\author{\authorblockN{Jaehyun Park\authorrefmark{1}, Byung Jang Jeong\authorrefmark{1}, and Yunju Park\authorrefmark{2}}
%\authorblockA{\authorrefmark{1}Broadcasting and
%Telecommunications Convergence Research Laboratory\\
%Electronics and Telecommunications Research Institute (ETRI),
%Daejeon, Korea\\
%Email: \{jhpark00, shwang, bjjeong\}@etri.re.kr}
%\authorblockA{\authorrefmark{2}Department of Mathematics and Computer Science\\
%Korea Science Academy of Korea Advanced Institute of Science and Technology (KAIST), Busan, Korea\\
%Email: yunjupark@kaist.ac.kr}}
% make the title area
\maketitle

\begin{abstract}
Recently, joint wireless information and energy
transfer (JWIET) methods have been proposed to relieve the battery limitation of wireless devices. However, the JWIET in a general K-user MIMO interference channel (IFC) has been unexplored so far. In this paper, we investigate for the first time the JWIET in K-user MIMO IFC, in which
receivers either decode the incoming information data (information
decoding, ID) or harvest the RF energy (energy harvesting, EH). In the
K-user IFC, we consider three different scenarios
according to the receiver mode -- i) multiple EH receivers and a single ID receiver, ii) multiple IDs and a single EH, and iii) multiple IDs and multiple EHs. For all scenarios, we have found a common necessary condition of the optimal
transmission strategy and, accordingly, developed the transmission strategy that satisfies the common necessary condition, in which all the transmitters transferring energy exploit a rank-one energy beamforming. Furthermore, we have also proposed an iterative algorithm to optimize the covariance matrices of the transmitters that transfer information and the powers of the energy beamforming transmitters simultaneously, and identified the corresponding achievable rate-energy tradeoff region. Finally, we have shown that by selecting EH receivers according to their signal-to-leakage-and-harvested energy-ratio (SLER), we can improve the achievable rate-energy region further.
%As the radio-frequency (RF) radiation has become a viable source for
%energy harvesting, joint wireless information and energy
%transfer (JWIET) methods have been proposed to relieve the battery limitation of mobile nodes. However, the JWIET in a general K-user MIMO interference channel (IFC) has been unexplored so far. In this paper, we first investigate the JWIET in K-user MIMO IFC, in which
%receivers either decode the incoming information data (information
%decoding, ID) or harvest the RF energy (energy harvesting, EH). In the
%K-user IFC, we consider three different scenarios
%according to the receiver mode -- i) multiple EH receivers and a single ID receiver, ii) multiple IDs and a single EH, and iii) multiple IDs and multiple EHs. For all scenarios, we have found a common necessary condition of the optimal
%transmission strategy and, accordingly, developed the transmission strategy that satisfies the common necessary condition, in which all the transmitters transferring energy exploit a rank-one energy beamforming. Furthermore, we have also proposed an iterative algorithm to optimize the covariance matrices of the transmitters that transfer information and the powers of the energy beamforming transmitters simultaneously, and identified the corresponding achievable rate-energy (R-E) tradeoff region. Finally, we have shown that by selecting EH receivers according to their signal-to-leakage-and-harvesting energy-ratio (SLER), which indicates how suitable the channel is to either EH mode or ID mode, we can improve the achievable rate-energy region further.
\end{abstract}


\section{Introduction}
\label{sec:intro}
One of the main challenges in modern wireless communication system is that wireless devices are resource-constrained, mainly due to battery limitation. Following the popularity of smart phones and various heavy-battery-consuming applications, 4th generation (4G) and beyond 4G standards also consider ways to address the battery limitation of wireless devices (e.g. device-to-device communications) \cite{3GPPMTC}. During the last decade, there has been a lot of
interest to transfer energy wirelessly and recently,
radio-frequency (RF) radiation has become a viable source for
energy harvesting. It is nowadays possible to transfer the energy
wirelessly with a reasonable efficiency over small distances and,
furthermore, wireless sensor networks (WSNs) in which the
sensors are capable of harvesting RF energy to power their own
transmissions have been introduced in the industry (\cite{Soljacic,
Yates, Vullers, Fiez} and references therein).

Until now, wireless energy transfer has been developed
independently from the wireless information transfer.
Interestingly, because RF signals carry information as well as
energy, ``joint wireless information and energy transfer (JWIET)'' has
attracted significant attention very recently \cite{Zhang1, Zhang2, KHuang1,
Ozel, RRajesh, RRajesh1, KIshibashi, ANasir, YLuo}. Previous works
have studied the fundamental performance limits and the optimal
transmission strategies of the JWIET in various communication scenarios such as the downlink of a cellular system with a single base
station (BS) and multiple mobile stations (MSs) \cite{KHuang1}, the cooperative relay system \cite{ANasir} and the broadcasting system \cite{Zhang1, Zhang2} with a single energy
receiver and a single information receiver when they are
separately located or co-located.
%For the energy
%harvesting transmitter \cite{Ozel,RRajesh, RRajesh1, KIshibashi, YLuo}, energy harvesting scheduling and transmit
%power allocation methods have been developed and, for the energy
%harvesting receiver \cite{Zhang1, Zhang2, KHuang1, ANasir}, the management of information decoding and
%energy harvesting has been developed.
Recently, considering multi-user MISO scenario, several transmission strategies and power allocation methods have been proposed \cite{XuZhang, ZhouZhangHo2, KwanSchober}. Furthermore, there have been several studies of JWIET in the interference channel (IFC)
\cite{Tutuncuoglu1, Tutuncuoglu2, KHuang2, ChenLiChang, ParkBruno}. In \cite{Tutuncuoglu1,
Tutuncuoglu2},  the optimal power scheduling at the energy harvesting transmitters are proposed for two-user single-input single-output (SISO) IFC such that the sum-rate is maximized for given harvested energy constraints. In \cite{KHuang2}, JWIET in multi-cell cellular networks is investigated, where all the BSs and MSs have a single antenna. In \cite{ChenLiChang}, by considering two-user single-input
multiple-output (SIMO) IFC, the system throughput is maximized subject to individual energy harvesting constraints and power constraints and extended it to K-user MISO IFC. Note that because the interference has different impacts on the performances of information decoding (ID) (negatively) and energy harvesting (EH) (positively) at the receivers, the transmission strategy for JWIET is a critical issue especially in IFC.
To the best of the authors¡¯ knowledge, JWIET in the general K-user MIMO IFC (which describes modern advanced communication systems) has not been addressed so far. Recently, in \cite{ParkBruno}, a JWIET in a two-user MIMO IFC has been studied and a necessary condition of the optimal transmission strategy for the two-user MIMO IFC has been derived. That is, in a two-user MIMO IFC, the energy transmitter may create a rank-one beam with the aim to either maximize the energy harvested at the EH receiver or minimize the interference at the ID receiver. Alternatively, it may generate multi-rank beams allocating its power on both directions. However, in \cite{ParkBruno}, it is proved that to achieve the optimal rate-energy (R-E) performance, the energy transmitter should take a rank-one beamforming strategy with a proper power control.

In this paper, we extend the results obtained in \cite{ParkBruno} and investigate JWIET in a K-user MIMO IFC, in which multiple MIMO transceiver pairs coexist and each receiver either decodes the incoming information data or harvests the RF energy. Throughout the paper, it is assumed that the receivers cannot perform ID and EH operations simultaneously, because existing circuits that harvest energy from the
received RF signal are not yet able to decode the information
carried through the same RF signal \cite{Zhang1, Zhang2,
ZhouZhangHo}. In \cite{Zhang1}, considering this practical issue, two different JWIET methods for MIMO broadcasting system - time switching and power splitting methods - have been proposed. In the time switching method, the receiver switches between ID mode and EH mode over time, while in the power splitting method, the received signal is split into two signals with different power that are the inputs of two disjoint ID and EH circuits, respectively. Because the power splitting method requires higher hardware complexity at the receivers (e.g., RF signal splitter), in this paper, we consider that each receiver switches between ID and EH modes in time-basis. Accordingly, we have three different scenarios
according to the receiver mode -- i) multiple EH receivers and a single ID receiver, ii) multiple IDs and a single EH, and iii) multiple IDs and multiple EHs.
%It is also assumed that all the transmitter have
%knowledge of their local CSI only, i.e. the CSI corresponding to
%the links between a transmitter and all receivers. Furthermore, the transmitters do not share the information data to be
%transmitted and the interference is assumed not decodable at the receivers as in \cite{Shen}.
For all scenarios, the optimal achievable R-E trade-off region is not easily identified and the optimal
transmission strategy is still unknown. However, in this paper, we have shown that the optimal energy transmitter's strategies for all three scenarios also become optimal for the {\it{properly-transformed}} two-user MIMO IFC. Therefore, we have found a common necessary condition of the optimal transmission strategy and developed the transmission strategy that satisfies the common necessary condition, in which all the transmitters transferring energy exploit a rank-one energy beamforming. Here, we have modified three different rank-one beamforming schemes, originally developed for two-user MIMO IFC \cite{ParkBruno} - maximum energy beamforming (MEB), minimum leakage beamforming (MLB), and signal-to-leakage-and-energy ratio (SLER) maximization beamforming, suitable to K-user MIMO IFC.
Given the rank-one beamforming at the energy transmitters, we have formulated the optimization problem for the achievable rate-energy region. However, because it is non-convex, we have proposed an iterative algorithm to optimize the covariance matrices of the transmitters that transfer information and the powers of the energy beamforming transmitters, simultaneously. We have shown that the powers of the energy beamforming transmitters converges monotonically, which guarantees the convergence of the proposed algorithm. In addition, when the number of energy transmitters increases, the ID receivers are affected by an increasing number of interfering beams (directions and power) that affect their information rate performance. This leads us to develop a new SLER maximizing beamforming with beam tilting. Here, the beam tilting means that we change the direction of an energy beam without changing its transmit power. Finally, we have proposed an efficient SLER-based EH transceiver selection method that further improves the achievable R-E region.

The rest of this paper is organized as follows. In Section
\ref{sec:systemmodel}, we introduce the system model for K-user
MIMO IFC. In Section \ref{sec:optimal_tx_strat}, we discuss the necessary condition for the optimal transmission strategies in the K-user MIMO IFC. In Section \ref{sec:rankoneBF_REregion}, we investigate the achievable R-E region for K-user MIMO IFC and, after formulating the optimization problem, propose an iterative algorithm to solve it.
In Section \ref{sec:simulation},
we provide several simulation results
and in Section \ref{sec:conc} we give our
conclusions.

Throughout the paper, matrices and vectors are represented by bold
capital letters and bold lower-case letters, respectively. The
notations $({\bf A})^{H} $, $({\bf A})^{\dagger} $, $({\bf A})_i$,
$[{\bf A}]_i$, $tr({\bf A})$, $\det({\bf A})$, and $\sigma_{k}({\bf A})$ denote the
conjugate transpose, pseudo-inverse, the $i$th row, the $i$th
column, the trace, the determinant, and the $k$th largest singular value of a matrix ${\bf A}$,
respectively. The matrix norm $\|{\bf A}\|$ and $\|{\bf A}\|_F$
denote the 2-norm and Frobenius norm of a matrix ${\bf A}$,
respectively, and the vector norm $\|{\bf a}\|$ denotes the 2-norm
of a vector ${\bf a}$. In addition, $(a)^+ \triangleq \max (a, 0)$
and ${\bf A} \succeq {\bf 0}$ means that a matrix ${\bf A}$ is positive
semi-definite. The matrix $diag\{{\bf A}_1,...,{\bf A}_M\}$ is a block diagonal matrix with block diagonal elements ${\bf A}_m$. Finally, ${\bf I}_{M}$ denotes the $M \times M$
identity matrix.

\section{System model}
\label{sec:systemmodel}

\begin{figure}
\begin{center}
\begin{tabular}{c}
\includegraphics[height=5.cm]{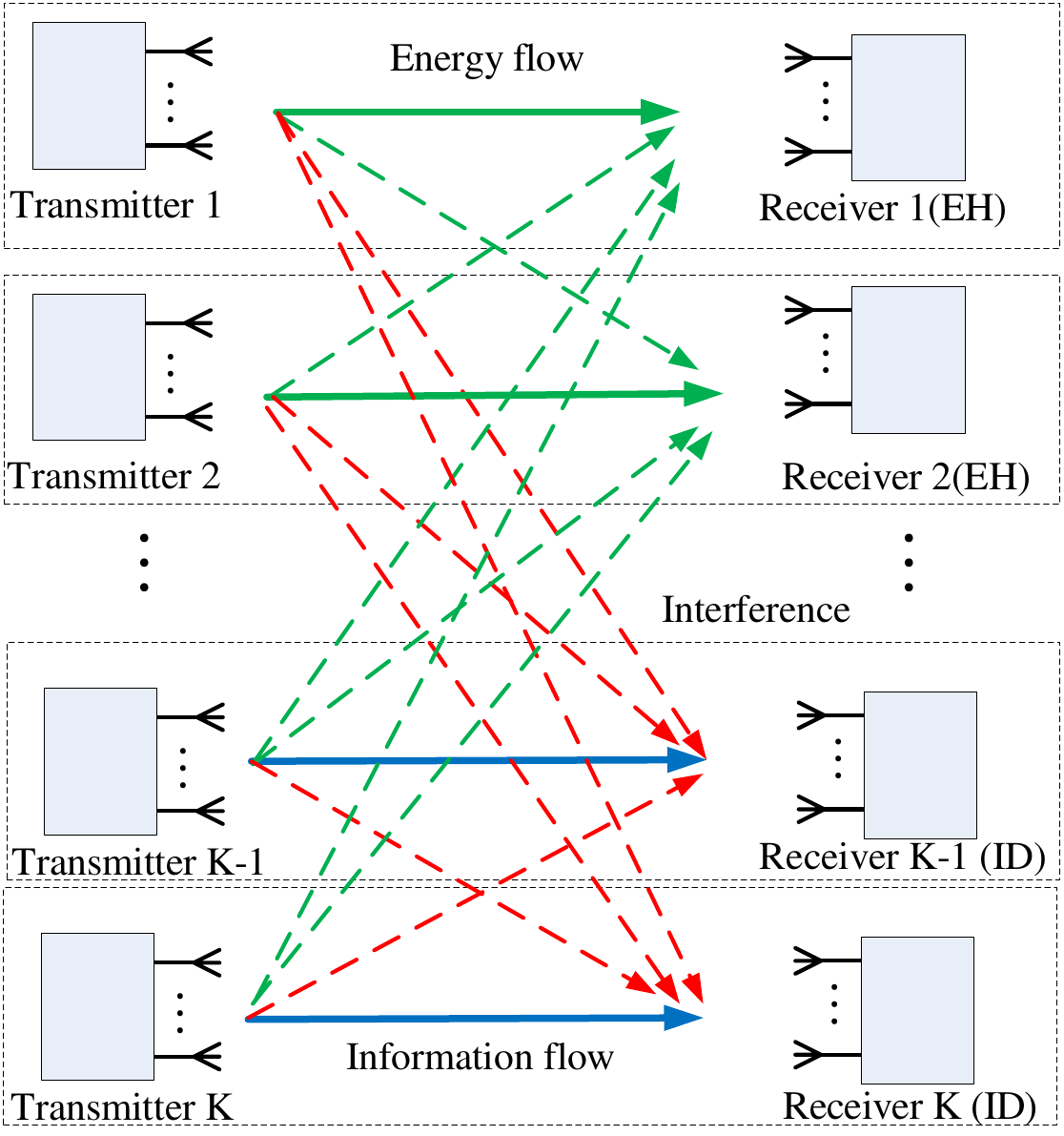}
\end{tabular}
\end{center}
\caption[JWIET_K_userIC_block]
%>>>> use \label inside caption to get Fig. number with \ref{}
{ \label{JWIET_K_userIC_block} K-user MIMO IFC.}
\end{figure}

We consider a K-user MIMO IFC system where $K$ transmitters,
each with $M_t$ antennas, are simultaneously transmitting their
signals to $K$ receivers, each with $M_r$ antennas, as shown in
Fig. \ref{JWIET_K_userIC_block}. Note that each receiver can
either decode the information or harvest energy from the received
signal, but it cannot execute the ID and EH operations at the same time due to the hardware limitations. That
is, each receiver can switch between ID
mode and EH mode at each frame or time slot. Here, the mode decided by the
receiver is also sent to all the transmitters through the zero-delay and
error-free feedback link at the beginning of the frame.
%\footnote{
%Note that the switching criterion between ID mode and EH mode
%depends on the receiver's condition such as the available energy
%in the storage and the required processing or circuit power. In
%this paper, we focus on the achievable rate and harvested energy
%obtained by the transferred signals from both transmitters in the
%IFC according to the different receiver modes. }
It is assumed that the transmitters have perfect knowledge of the CSI of their
associated links (i.e. the links between a transmitter and all
receivers) but do not share those CSI among the transmitters.
Furthermore, $M_t = M_r = M$ for simplicity, but it can be
extended to general antenna configurations. Assuming a frequency
flat fading channel, which is static over several frames, the
received signal ${\bf y}_i \in \mathbb{C}^{M \times 1}$ for $i=1,...,
K$ can be written as
\begin{eqnarray}\label{Sys_1}
{\bf y}_i = \sum_{j=1}^{K}{\bf H}_{ij} {\bf x}_j + {\bf
n}_i,
\end{eqnarray}
where ${\bf n}_i \in \mathbb{C}^{M\times 1} $ is a complex white
Gaussian noise vector with a covariance matrix $\sigma_n^2{\bf
I}_{M}$ and ${\bf H}_{ij} \in \mathbb{C}^{M\times M}$ is the
normalized frequency-flat fading channel from the $j$th
transmitter to the $i$th receiver such as $\sum_{l,k= 1}^{M}
|h_{ij}^{(l,k)}|^2 = \alpha_{ij}M$ \cite{Loyka}. Here,
$h_{ij}^{(l,k)}$ is the $(l, k)$th element of ${\bf H}_{ij}$ and
$\alpha_{ij}\in [0,1]$. We
assume that $ {\bf H}_{ij}$ has a full rank. %with a high
%probability ($\approx 1$). %${\bf H}_{ij} \in \mathbb{C}^{M\times M}$ is the channel
%matrix from the $j$th transmitter to the $i$th receiver whose
%elements are independent and identically distributed (i.i.d.)
%zero-mean complex Gaussian random variables (RVs) with a unit
%variance
The vector ${\bf x}_j \in \mathbb{C}^{M \times 1}$ is the transmit
signal, in which the independent messages can be conveyed, at the
$j$th transmitter with a transmit power constraint for $j= 1,...,K$ as
\begin{eqnarray}\label{Sys_2}
E[\|{\bf x}_j\|^2] \leq P  {\text{ for }} j=1,...,K.
\end{eqnarray}
In this paper, the SNR measured at the $i$th receiver is %, in which the
%power of channel and the transmit power are aggregated,
defined as $ SNR_i = \frac{E[\|H_{ii}\|_F^2 \|x\|^2]}{E[\|n\|^2]}
= \frac{\alpha_{ii}P}{\sigma_n^2}$. Throughout the paper, to ease
readability, it is assumed without loss of generality that
$\sigma_n^2 = 1$, unless otherwise stated. General environments,
characterized by other values of the channel/noise power, can be
described simply by adjusting $P$.

When the receiver operates in ID mode, the achievable rate at
$i$th receiver, $R_i$, is given by \cite{Scutari}
\begin{eqnarray}\label{Sys_3_revised}
R_i = \log \det ({\bf I}_{M} +{\bf H}_{ii}^H{\bf R}_{-i}^{-1}{\bf
H}_{ii}{\bf Q}_i ),
\end{eqnarray}
where ${\bf R}_{-i}$ indicates the covariance matrix of noise and
interference at the $i$th receiver, i.e.,
\begin{eqnarray}\nonumber
{\bf R}_{-i} = {\bf I}_{M} + \sum_{j\neq i}^{K} {\bf H}_{ij}{\bf Q}_j {\bf
H}_{ij}^H.
\end{eqnarray}
%where ${\bf R}_{-i} = {\bf I}_{M} + {\bf H}_{ij}{\bf Q}_j {\bf
%H}_{ij}^H$ with $j \neq i$ indicates the covariance matrix of
%noise and interference.
Here, ${\bf Q}_j = E [{\bf x}_j {\bf x}_j^H ]$ denotes the
covariance matrix of the transmit signal at the $j$th transmitter
and, from (\ref{Sys_2}), $tr({\bf Q}_j) \leq P $.

For EH mode, it can be assumed that the total harvested power
$E_i$ at the $i$th receiver (more exactly, harvested energy
normalized by the baseband symbol period) is given by
\begin{eqnarray}\label{Sys_3}
E_i &=& \zeta_i E[\|{\bf y}_i \|^2 ]~=~ \zeta_i tr\left( \sum_{j=1}^K{\bf H}_{ij}{\bf Q}_j{\bf
H}_{ij}^H +{\bf I}_{M} \right),
\end{eqnarray}
where $\zeta_i$ denotes the efficiency constant for converting the
harvested energy to electrical energy to be stored \cite{Vullers,
Zhang1}. For simplicity, it is assumed that $\zeta_i=1$ and the
noise power is negligible compared to the transferred energy from
each transmitters. \footnote{In this paper, we assume the system
operates in the high signal-to-noise ratio (SNR) regime, which is
also consistent with a practical wireless energy transfer that
requires a high-power transmission.} %, but we also discuss the low
%SNR regime in Section \ref{sec:discussion} as well.}
That is,
\begin{eqnarray}\label{Sys_4}E_i &\approx& tr\left( \sum_{j=1}^K{\bf
H}_{ij}{\bf Q}_j{\bf H}_{ij}^H \right) ~=~  \sum_{j=1}^K tr\left({\bf
H}_{ij}{\bf Q}_j{\bf H}_{ij}^H\right),
\end{eqnarray}
where $E_{ij}=tr\left( {\bf H}_{ij}{\bf Q}_j{\bf H}_{ij}^H
\right)$ denotes the energy transferred from the $j$th
transmitter to the $i$th receiver.

%Note that, because each receiver cannot decode the information and
%simultaneously harvest energy, there are four possible receiving
%modes for the two-user MIMO IC -- ($EH_1$, $EH_2$), ($ID_1$,
%$ID_2$), ($EH_1$, $ID_2$), ($ID_1$, $EH_2$).
Note that, when the receiver decodes the information data from
the associated transmitter under the assumption that the signals
from the other transmitters are not decodable \cite{Shen}, the
signals from the other transmitters become an interference to be
defeated. In contrast, when the receiver harvests the energy, they
become a useful energy-transferring source. In Fig.
\ref{JWIET_K_userIC_block}, the interference denoted by the dashed
red line should be reduced for IDs, while the interference by the dashed green line be maximized for EHs.

\section{A necessary condition for the optimal transmission strategy}\label{sec:optimal_tx_strat}

In \cite{ParkBruno}, a necessary condition of the optimal transmission strategy has been found for the two-user MIMO IFC with one EH and one ID, in which the energy transmitter should take a rank-one energy beamforming strategy with a proper power control. In this section, we first review one EH and one ID in a two-user MIMO IFC, briefly. Then, we will look into the cases of one ID/EH and multiple EHs/IDs. Finally, we consider the case of multiple IDs and multiple EHs.

%The following lemma gives a useful insight into the derivation of the optimal boundary.
%\begin{lem}\label{lem1} For ${\bf H}_{11}$ and ${\bf H}_{21}$,
%there always exists an invertible matrix ${\bf T}\in
%\mathbb{C}^{M\times M}$ such that
%\begin{eqnarray}\label{oneIDoneEH2}
%{\bf U}_{G}^H{\bf H}_{11}{\bf T} = {\bf \Sigma}_G\nonumber \\
%{\bf V}_{G}^H{\bf H}_{21}{\bf T} = {\bf I}_{M},
%\end{eqnarray}
%where ${\bf U}_{G}$ and ${\bf V}_{G}$ are unitary and ${\bf
%\Sigma}_G$ is a diagonal matrix with $ \sigma_{G,1} \geq
%\sigma_{G,2} \geq ,...,\geq \sigma_{G,M}\geq 0$.
%\end{lem}
%\begin{proof}
%Because $ {\bf H}_{21}$ has a full rank, by utilizing the
%generalized singular value decomposition \cite{Sadek, Paige}, we
%can obtain an invertible matrix ${\bf T}'$ such that
%\begin{eqnarray}\label{oneIDoneEH2_1}
%{\bf U}_{G}^H{\bf H}_{11}{\bf T}' = {\bf \Sigma}_A\nonumber \\
%{\bf V}_{G}^H{\bf H}_{21}{\bf T}' = {\bf \Sigma}_{B},\nonumber
%\end{eqnarray}
%where ${\bf U}_{G}$ and ${\bf V}_{G}$ are unitary and ${\bf
%\Sigma}_A$ and ${\bf \Sigma}_B$ are diagonal matrices with $1 \geq
%\sigma_{A,1} \geq \sigma_{A,2} \geq ,...,\geq \sigma_{A,M}\geq 0$
%and with $0 < \sigma_{B,1} \leq \sigma_{B,2} \leq ,...,\leq
%\sigma_{B,M}\leq 1$, respectively. Here, $\sigma_{A,i}^2
%+\sigma_{B,i}^2 = 1$. Therefore, by setting ${\bf T} ={\bf T}'{\bf
%\Sigma}_{B}^{-1}$, we can obtain (\ref{oneIDoneEH2}) with ${\bf
%\Sigma}_G = {\bf \Sigma}_A{\bf \Sigma}_{B}^{-1}$.
%\end{proof}
%Note Lemma 1 is more general version of Lemma in [].

\subsection{One ID receiver and One EH receiver}\label{ssec:oneIDoneEH}

In this subsection, without loss of generality, we consider the transceiver pair ($Tx_1$, $Rx_1$) operates in EH mode,
while ($Tx_2$, $Rx_2$) in ID mode. Because information decoding is done only at the second receiver,
by letting $R= R_2$ and $E = E_1 =E_{11}+E_{12}$, we can define
the achievable rate-energy region as:
\begin{eqnarray}\label{oneIDoneEH_1}
\!\!\!&\!\!C_{R\!-\!E} (P) \!\triangleq  \!\Biggl\{ \!(R, E)\! : R \!\leq\!
\log \!\det({\bf I}_{M} \!+\! {\bf H}_{22}^H{\bf R}_{-2}^{-1}{\bf
H}_{22}{\bf Q}_2 ),\!E \!\leq \!\displaystyle\sum_{\!j\!=\!1}^{\!2} tr ({\bf H}_{1j} {\bf
Q}_j {\bf H}_{1j}^H), tr({\bf Q}_j)\!\leq \!P, {\bf Q}_j\!\succeq
\!{\bf 0}, j\!=\!1,\!2\! \Biggr\}\!.\!&\!
\end{eqnarray}
The following proposition tells about the rank-one optimality in the two-user MIMO IFC.
\begin{prop}\label{prop1} In the high SNR regime, the optimal ${\bf Q}_1$ at the boundary
of the achievable rate-energy region has a rank one at most. That
is, $rank ({\bf Q}_1) \leq 1$.
\end{prop}
\begin{proof}
The detailed proof is given in \cite{ParkBruno}, but here is its brief sketch. If the energy $\bar E$ at the boundary point of the achievable rate-energy is small enough such that $\bar E \leq tr({\bf H}_{12}{\bf Q}_2{\bf H}_{12}^H)$, then $rank({\bf Q}_1) = 0$ (i.e., the first transmitter does not need to transmit any signal causing the interference to the ID receiver).

If $\bar E > tr({\bf H}_{12}{\bf Q}_2{\bf H}_{12}^H)$, we then assume $m = rank({\bf Q}_1)\geq 1$. Based on the generalized singular value decomposition (GSVD) of (${\bf H}_{11}$, ${\bf H}_{21}$) and the interlacing theorem (Theorem 3.1 in \cite{RHorn}), with ${\bf Q}_1$ satisfying the required harvesting energy, the achievable rate $\bar R$ at high SNR can be approximated as \cite{ParkBruno}
\begin{eqnarray}\label{oneIDoneEH_1_rev}
{\bar R} \approx f({\bf H}_{22}) -  \log\left({
\prod_{i=1}^{m}(1+\sigma_{x,i}^2)}\right),
\end{eqnarray}
where $\sigma_{x,i}$ is the singular values of an arbitrary matrix ${\bf X}$ with ${\bf Q}_1 = {\bf T}{\bf X}{\bf X}^H {\bf T}^H$. Here, ${\bf T}$ is an invertible matrix obtained from GSVD of (${\bf H}_{11}$, ${\bf H}_{21}$) and
$\sigma_{x,1}\geq,...,\geq\sigma_{x,m}\geq 0$ such that $\sum_{i=1}^m \alpha_i \sigma_{x,i}^2 =\bar E_{11}$ with $\alpha_{1}\geq,...,\geq\alpha_{m}\geq 0$ and a fixed constant $\bar E_{11}$. Then, we can easily find that $\bar R$ is maximized when $m=rank({\bf Q}_1)=1$.
\end{proof}

From Proposition \ref{prop1}, when transferring the energy in
the IFC, the energy transmitter's optimal strategy is either a rank-one
beamforming or no transmission according to the energy harvested
from the information transmitter. Such strategy increases the harvested energy
at the corresponding EH receiver and simultaneously reduces the
interference at the other ID receiver. Intuitively, from the power transfer point of view, ${\bf Q}_1$
should be as close to the dominant eigenvector of ${\bf
H}_{11}^H{\bf H}_{11}$ as possible, which implies that the rank
one is optimal for power transfer. From the information transfer
point of view, when SNR goes to infinity, the rate maximization is
equivalent to the DOF maximization. That is, a larger rank for
${\bf Q}_1$ means that more dimensions at the second receiver will
be interfered. Therefore, a rank one for ${\bf Q}_1$ is optimal
for both information and power transfer. Note that Proposition \ref{prop1} is based on the high SNR regime, but the rank-one optimality is also valid in the low SNR regime as discussed in Section VI.A of \cite{ParkBruno}.

%%%%%%%%%%%%%%%%%%%55
%%%%%%%%%%%%%%%%%%%55
%% Notation on channel matrix should be provided
%%%%%%%%%%%%%%%%%%%55
%%%%%%%%%%%%%%%%%%%55

\subsection{One ID and multiple EHs in a K-user IFC}\label{ssec:oneIDmultiEH}
Without loss of generality, the transceiver pairs $(Tx_k, Rx_k)$, $k=1,...,K-1$ operate in EH mode, while $(Tx_K, Rx_K)$ in ID mode.
Note that energy transmitters optimize their transmission strategies in a distributed manner.

Because information decoding is done only at $Rx_K$,
by letting\footnote{To consider different priorities for either energy or rate, the weighted sum-rate or weighted sum-energy can be used as the objective functions (similarly done for the information transfer only \cite{QShi}) and inspired by \cite{QShi}, our current approaches can be extended to the weighted objective functions.} $R= R_K$ and $E = \sum_{i=1}^{K-1}E_i$ with $E_i = \sum_{j=1}^{K} E_{ij}$, we can define
the achievable rate-energy region as:
\begin{eqnarray}\label{oneIDmultiEH_1}
\!\!&\!C_{R\!-\!E} (P) \!\triangleq  \!\Biggl\{ \!(R, E) : R \leq
\log \det({\bf I}_{M} + {\bf H}_{KK}^H{\bf R}_{-K}^{-1}{\bf
H}_{KK}{\bf Q}_K ),\!&\!\nonumber\\
\!\!&\!E \!\leq \!\sum_{\!i\!=\!1}^{\!K-1} \sum_{\!j\!=\!1}^{\!K} tr ({\bf H}_{ij} {\bf
Q}_j {\bf H}_{ij}^H), tr({\bf Q}_j)\!\leq \!P, {\bf Q}_j\!\succeq
\!{\bf 0}, j\!=\!1,...,\!K\! \Biggr\}\!.\!&\!
\end{eqnarray}
Then, we have the following proposition.
\begin{prop}\label{prop_oneIDmultiEH} All the optimal ${\bf Q}_k$ at the boundary
of the achievable rate-energy region for (\ref{oneIDmultiEH_1}) become optimal solutions for the boundary of
\begin{eqnarray}\label{oneIDmultiEH_2}
\!\!&\!C_{R\!-\!E,k} (P) \!\triangleq  \!\Biggl\{ \!(R, E) : R \leq
\log \det({\bf I}_{M} + \tilde{\bf H}_{22}^H({\bf R}_{-2}^{(k)})^{-1}\tilde{\bf H}_{22}\tilde{\bf Q}_2 ),\!&\!\nonumber\\
\!\!&\!E \!\leq \! tr (\tilde{\bf H}_{11}^{(k)} {\bf
Q}_k (\tilde{\bf H}_{11}^{(k)})^H) + tr (\tilde{\bf H}_{12} {\bf
Q}_K \tilde{\bf H}_{12}^H) , tr({\bf Q}_j)\!\leq \!P, {\bf Q}_j\!\succeq
\!{\bf 0},  j\!=\!k,K\!\Biggr\}\!\!&\!
\end{eqnarray}
for all $k=1,...,K-1$, %, and vice versa
where
\begin{eqnarray}\label{oneIDmultiEH_4}
\tilde{\bf H}_{11}^{(k)} = \left[\begin{array}{c}{\bf H}_{1k} \\ \vdots \\ {\bf H}_{K\!-\!1k}  \end{array}\right], \quad \tilde{\bf H}_{12} = \left[\begin{array}{c}{\bf H}_{1K} \\ \vdots \\ {\bf H}_{K\!-\!1K}  \end{array}\right] \in \mathbb{C}^{(K-1)M\times M},
\end{eqnarray}
\begin{eqnarray}\label{oneIDmultiEH_3}
\tilde{\bf H}_{21}^{(k)} = {\bf H}_{Kk}, ~~\tilde{\bf H}_{22}= {\bf H}_{KK},
\end{eqnarray}
and
\begin{eqnarray}\label{oneIDmultiEH_3_1}
\tilde{\bf Q}_{2}= {\bf Q}_{K},~~{\bf R}_{-2}^{(k)}= {\bf I}_{M} +  {\bf C}_{\bar k} + \tilde{\bf H}_{21}^{(k)}{\bf Q}_k (\tilde{\bf H}_{21}^{(k)})^H.
\end{eqnarray}
Here,  ${\bf C}_{\bar k}$ is the covariance matrix of the interference from the other energy transmitters given as ${\bf C}_{\bar k} =  \sum_{\!j\!=\!1, j\!\neq\! k}^{\!K-1}   \tilde{\bf H}_{21}^{(j)} {\bf
Q}_j (\tilde{\bf H}_{21}^{(j)})^H$.
%That is, $C_k {\bf Q}_k$ with a constant scaling factor $C_k$ yields the boundary point of (\ref{oneIDmultiEH_2}).
\end{prop}
\begin{proof}
%Let us consider the boundary point $(\bar R, \bar E)$ of the achievable rate-energy for (\ref{oneIDmultiEH_1}) with any $\bar{\bf Q}_K$, the covariance matrix of the information transmitter on the boundary point. In addition, the corresponding covariance matrices of the energy transmitters at the boundary point are denoted as $\bar{\bf Q}_j$, $j=1,...,K-1$. Note that, at the boundary point, the perturbation on $\bar{\bf Q}_j$ for any $j \in \{1,...,K\}$ does not increase both $\bar R$ and $\bar E$ values, simultaneously. That is, the increase of $\bar R$ ($\bar E$) due to the variation on $\bar{\bf Q}_j$ always induces the decrease of $\bar E$ ($\bar R$) at the boundary.
Note that ${\bf R}_{-K}$ in (\ref{oneIDmultiEH_1}) can be rewritten as
\begin{eqnarray}\label{oneIDmultiEH_5_rev2_1}
&{\bf R}_{-K}={\bf I}_M + \sum_{j\neq K, k}{\bf H}_{Kj}\bar{\bf Q}_j{\bf H}_{Kj}^H +{\bf H}_{Kk}{\bf Q}_k{\bf H}_{Kk}^H,&
\end{eqnarray}
and is exactly the same as ${\bf R}_{-2}^{(k)}$ in (\ref{oneIDmultiEH_3_1}). Accordingly, the rate at the boundary of (\ref{oneIDmultiEH_1}) can be rewritten as that of (\ref{oneIDmultiEH_2}). Furthermore, in (\ref{oneIDmultiEH_1}),
\begin{eqnarray}\label{oneIDmultiEH_7}
\!\!\sum_{\!i\!=\!1}^{\!K-1} \sum_{\!j\!=\!1}^{\!K} tr ({\bf H}_{ij} {\bf
Q}_j {\bf H}_{ij}^H) = tr (\tilde{\bf H}_{11}^{(k)} {\bf
Q}_k (\tilde{\bf H}_{11}^{(k)})^H) + \sum_{\!j\!=\!1, j\!\neq\! k}^{\!K\!-\!1}\!   tr (\tilde{\bf H}_{11}^{(\!j\!)}\! {\bf
Q}_j (\tilde{\bf H}_{11}^{(\!j\!)}\!)^H\!) \!+\!   tr (\tilde{\bf H}_{12} \!{\bf
Q}_K\! \tilde{\bf H}_{12}^H\!),\!\!
\end{eqnarray}
and $\sum_{\!j\!=\!1, j\!\neq\! k}^{\!K-1}   tr (\tilde{\bf H}_{11}^{(j)} {\bf
Q}_j (\tilde{\bf H}_{11}^{(j)})^H) $ shifts the trade-off curve of (\ref{oneIDmultiEH_1}) along the $E$-axis as in Fig. \ref{figK_user_OneID}. Therefore, all the optimal $\bar{\bf Q}_k$ for all boundary points $(\bar R, \bar E)$ of (\ref{oneIDmultiEH_1}) become solutions for the boundary of (\ref{oneIDmultiEH_2}), even though the boundaries of (\ref{oneIDmultiEH_1}) and (\ref{oneIDmultiEH_2}) are different.
\end{proof}

\begin{figure}
\begin{center}
\begin{tabular}{c}
\includegraphics[height=4.6cm]{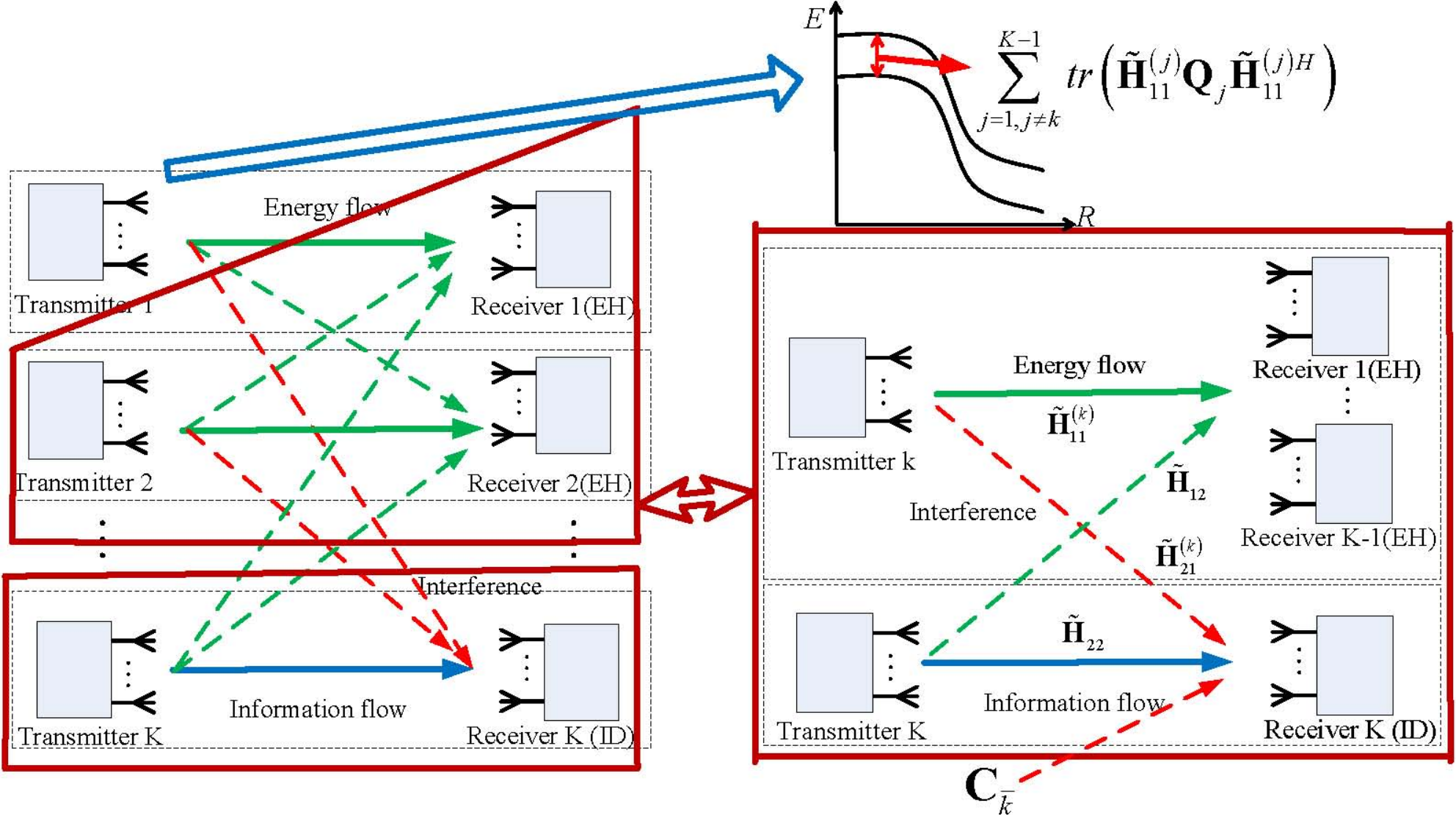}
\end{tabular}
\end{center}
\caption[K_user_OneID]
%>>>> use \label inside caption to get Fig. number with \ref{}
{ \label{figK_user_OneID} One ID and multiple EHs in a K-user IFC.}
\end{figure}

\begin{remark}\label{remark_oneIDmultiEH}
Note that (\ref{oneIDmultiEH_2}) can be regarded as the achievable R-E tradeoff region for the two-user MIMO IFC with an effective channel set of ($\tilde{\bf H}_{11}^{(k)}$, $\tilde{\bf H}_{12}$, $\tilde{\bf H}_{21}^{(k)}$,  $\tilde{\bf H}_{22}$) as in Fig. \ref{figK_user_OneID}. But, compared to the conventional two-user MIMO IFC, the effective two-user MIMO IFC for (\ref{oneIDmultiEH_2}) has the interference from other energy transmitters to the ID receiver which is unknown to the $k$th energy transmitter. The following lemma gives an important insight into the necessary condition for the optimal boundary for one ID and multiple EHs.
\begin{lem}\label{lem3}(\cite{ParkClerckxJSAC}, Lemma 1) For a positive semi-definite matrix ${\bf X}$ and a positive definite matrix ${\bf S}$, let
\begin{eqnarray}\label{lem3_1}
f({\bf X}) \triangleq \log \det({\bf I}_{M} + {\bf S}({\bf I}_M + {\bf X})^{-1} ).
\end{eqnarray}
Then, the maximization of $f({\bf X})$ with respect to ${\bf X}$ is equivalent with the minimization of $\det({\bf I}_{M} +{\bf X})$ with respect to ${\bf X}$.
\end{lem}

Accordingly, we can have the following important corollary.
\begin{cor}\label{cor_oneIDmultiEH} (One ID and multiple EHs) For $k=1,...,K-1$, the optimal ${\bf Q}_k$ at the boundary of the achievable rate-energy region (\ref{oneIDmultiEH_1}) has a rank one at most. That
is, $rank ({\bf Q}_k) \leq 1$ for $k=1,...,K-1$.
\end{cor}
\begin{proof}
First we show that in the effective two user-MIMO IFC with any external interference to the ID receiver that is unknown to the energy transmitter, the optimal ${\bf Q}_k$ of the $k$th energy transmitter at the boundary of the achievable rate-energy region (\ref{oneIDmultiEH_2}) has a rank one at most.

Let us consider the boundary point $(\bar R, \bar E)$ of the achievable rate-energy for (\ref{oneIDmultiEH_2}) with any $\tilde{\bf Q}_2$, the covariance matrix of the information transmitter and $\bar{\bf Q}_j$, $j\neq k$, the covariance matrices of other energy transmitters on the boundary point. Furthermore, let there be ${\bf Q}_k \triangleq {\bf U}_k{\bf \Sigma}_k{\bf U}_k^H$ with $m = rank({\bf Q}_k)>1 $, ${\bf \Sigma}_k = diag \{ P_{k1},...,P_{km}\}$, ${\bf U}_k^H{\bf U}_k ={\bf I}_{m}$, and $\sum_{i=1}^{m}P_{ki} = P_k$ which corresponds to the boundary point $(\bar R, \bar E)$. Then, given the harvested energy $\bar E$ (the boundary point) and $\tilde{\bf Q}_2$, the covariance matrix ${\bf Q}_k$ exhibits
\begin{eqnarray}\label{oneIDmultiEH_2rev2}
\bar R &=& \log \det({\bf I}_{M} + \tilde{\bf H}_{22}\tilde{\bf Q}_2\tilde{\bf H}_{22}^H({\bf I}_{M} + {\bf C}_{\bar k}+ \tilde{\bf H}_{21}^{(k)}{\bf Q}_k (\tilde{\bf H}_{21}^{(k)})^H)^{-1} ),\nonumber\\
\bar E_{11}^{(k)} &=&  tr (\tilde{\bf H}_{11}^{(k)} {\bf Q}_k (\tilde{\bf H}_{11}^{(k)})^H),
\end{eqnarray}
where $\bar E_{11}^{(k)}\triangleq \bar E -  tr (\tilde{\bf H}_{12} \tilde{\bf
Q}_2 \tilde{\bf H}_{12}^H) $. From Lemma \ref{lem3}, at the boundary point, $\det({\bf I}_{M} +  {\bf C}_{\bar k}+ \tilde{\bf H}_{21}^{(k)}{\bf Q}_k (\tilde{\bf H}_{21}^{(k)})^H)$ is minimized. Because ${\bf C}_{\bar k}$ has the EVD as ${\bf C}_{\bar k} = {\bf U}_{C}{\bf \Sigma}_C {\bf U}_{C}^H$, we can have
\begin{eqnarray}\label{oneIDmultiEH_2rev3}
&\det({\bf I}_{M} +  {\bf C}_{\bar k}+ \tilde{\bf H}_{21}^{(k)}{\bf Q}_k (\tilde{\bf H}_{21}^{(k)})^H) = \det({\bf I}_{M} +  {\bf \Sigma}_{C}+ {\bf U}_{C}^H\tilde{\bf H}_{21}^{(k)}{\bf Q}_k (\tilde{\bf H}_{21}^{(k)})^H{\bf U}_{C})&\nonumber\\
&=\det({\bf I}_{M} +  {\bf \Sigma}_{C})\det({\bf I}_M+ ({\bf I}_{M} +  {\bf \Sigma}_{C})^{-1}{\bf U}_{C}^H\tilde{\bf H}_{21}^{(k)}{\bf Q}_k (\tilde{\bf H}_{21}^{(k)})^H{\bf U}_{C}) &\nonumber\\
&= \det({\bf I}_{M} +  {\bf \Sigma}_{C})\det({\bf I}_{m} + {\bf \Sigma}_k^{1/2} {\bf U}_k^H(\tilde{\bf H}_{21}^{(k)})^H{\bf U}_{C}({\bf I}_{M} +  {\bf \Sigma}_{C})^{-1}{\bf U}_{C}^H\tilde{\bf H}_{21}^{(k)}{\bf U}_k{\bf \Sigma}_k^{1/2} ).&
\end{eqnarray}
Then, the first determinant in the right-hand-side of (\ref{oneIDmultiEH_2rev3}) is independent of ${\bf Q}_k$ and, in addition, we can easily find that the second determinant is minimized when
\begin{eqnarray}\label{oneIDmultiEH_2rev4}
[{\bf U}_k]_1 = {\bf v}_I, \quad  P_{k1}= P_k,~ P_{ki}= 0 {\text{ for }} i\neq 1.
\end{eqnarray}
where $ {\bf v}_I$ is the right singular vector associated with the smallest singular value of $(\tilde{\bf H}_{21}^{(k)})^H{\bf U}_{C}({\bf I}_{M} +  {\bf \Sigma}_{C})^{-1}{\bf U}_{C}^H\tilde{\bf H}_{21}^{(k)}$. Furthermore, by letting ${\bf v}_E$ as the right singular vector associated with the largest singular value of $ \tilde{\bf H}_{11}^{(k)}$, then $\bar E_{11}^{(k)} $ in (\ref{oneIDmultiEH_2rev2}) is maximized when
\begin{eqnarray}\label{oneIDmultiEH_2rev5}
[{\bf U}_k]_1 = {\bf v}_E, \quad  P_{k1}= P_k,~ P_{ki}= 0 {\text{ for }} i\neq 1.
\end{eqnarray}
Without loss of generality, ${\bf Q}_k$ can be defined as
\begin{eqnarray}\label{oneIDmultiEH_2rev6}
{\bf Q}_k = \sum_{i=1}^{m}{\bf u}_i{\bf u}_i^H P_{ki},
\end{eqnarray}
by choosing ${\bf u}_1$ such that it is in the range space of $[{\bf v}_I {\bf v}_E]$ and ${\bf u}_i^H{\bf u}_j = \left\{\begin{array}{c} 0{\text{ for }}i\neq j \\ 1{\text{ for }}i = j\end{array} \right. $. That is, ${\bf u}_1 = {\bf P}_{[{\bf v}_I {\bf v}_E]}{\bf u}_1$, where ${\bf P}_{[{\bf v}_I {\bf v}_E]}$ is a projection matrix onto the range space of $[{\bf v}_I {\bf v}_E]$. Therefore, if there exists $m >1 $ such that (\ref{oneIDmultiEH_2rev2}) is satisfied, we can always find $m'=1$ such that ${\bar E}_{11}'^{(k)} \geq \bar E_{11}^{(k)}$ and $\bar R' \geq \bar R$ with $P_{k1}= P_k$ $P_{ki}= 0$ for $i\neq 1$.

Now we are ready to show the corollary. Assuming that $rank({\bf Q}_k)\geq 2$ with some $k$ for (\ref{oneIDmultiEH_1}). From Proposition \ref{prop_oneIDmultiEH}, ${\bf Q}_k$ then becomes a solution for the boundary of (\ref{oneIDmultiEH_2}). However, from the above observation, the optimal ${\bf Q}_i$ at the boundary of the achievable R-E region of (\ref{oneIDmultiEH_2}) has a rank one at most, which is a contradiction.
\end{proof}
In other words, if a covariance matrix of an energy transmitter in the K-user MIMO IFC (multiple EHs and one ID) has a rank ($\geq 2$), then we can always find a rank-one beamforming for that transmitter exhibiting either higher information rate or larger harvested energy. Interestingly, when the interference unknown to the energy transmitter is added to the ID receiver in the two-user MIMO IFC, the rank-one optimality still holds. This is a generalized version of Proposition \ref{prop1}, but from (\ref{oneIDmultiEH_2rev4}), the optimal beamforming direction depends on the covariance matrix of the interference from other energy transmitters (specifically, the beamforming directions of the other energy transmitters).
%%%%%%%%%%%%%%%%%%%%%%%%%%%%%%%5
%power allocation issue......
%%%%%%%%%%%%%%%%%%%%%%%%%%%%%%%%%
\end{remark}

\subsection{One EH and multiple IDs in a K-user IFC}\label{sec:single}

Without loss of generality, the transceiver pair $(Tx_1, Rx_1)$ operate in EH mode, while $(Tx_k, Rx_k)$, $k=2,...,K$ in ID mode.

Because information decoding is done only at $Rx_i$, $i=2,...,K$
by letting $R=\sum_{i=2}^{K} R_i$ and $E = E_1$ with $E_1 = \sum_{j=1}^{K} E_{1j}$, we can define
the achievable rate-energy region as:
\begin{eqnarray}\label{multiIDoneEH_1}
\!\!&\!C_{R\!-\!E} (P) \!\triangleq  \!\Biggl\{ \!(R, E) : R \leq
\sum_{i=2}^{K}\log \det({\bf I}_{M} + {\bf H}_{ii}^H{\bf R}_{-i}^{-1}{\bf
H}_{ii}{\bf Q}_i ),\!&\!\nonumber\\
\!\!&\!E \!\leq \! \sum_{\!j\!=\!1}^{\!K} tr ({\bf H}_{1j} {\bf
Q}_j {\bf H}_{1j}^H), tr({\bf Q}_j)\!\leq \!P, {\bf Q}_j\!\succeq
\!{\bf 0}, j\!=\!1,...,\!K\! \Biggr\}\!.\!&\!
\end{eqnarray}

\begin{prop}\label{prop_multiIDoneEH} All the optimal ${\bf Q}_1$ at the boundary
of the achievable rate-energy region for (\ref{multiIDoneEH_1}) become optimal solutions for the boundary of
\begin{eqnarray}\label{multiIDoneEH_4}
\!\!&\!C_{R\!-\!E} (P) \!\triangleq  \!\Biggl\{ \!(R, E) : R \leq
\log \det({\bf I}_{(K-1)M} + \tilde{\bf H}_{22}^H(\tilde{\bf R}_{-2}^{(1)})^{-1}\tilde{\bf
H}_{22}\tilde{\bf Q}_2 ),\!&\!\nonumber\\
\!\!&\!E \!\leq \! tr (\tilde{\bf H}_{11}^{(1)} {\bf
Q}_1 (\tilde{\bf H}_{11}^{(1)})^H) + tr (\tilde{\bf H}_{12} \tilde{\bf
Q}_2 \tilde{\bf H}_{12}^H) , tr({\bf Q}_j)\!\leq \!P, {\bf Q}_j\!\succeq
\!{\bf 0},  j\!=\!1,...,K\!\Biggr\}\!,\!&\!
\end{eqnarray}
where
\begin{eqnarray}\label{multiIDoneEH_5}\nonumber
\tilde{\bf H}_{11}^{(1)} ={\bf H}_{11}, ~~ \tilde{\bf H}_{12} = \left[\begin{array}{ccc}{\bf H}_{12} & \hdots & {\bf H}_{1K}  \end{array}\right] \in \mathbb{C}^{M\times (K-1)M},
\end{eqnarray}
\begin{eqnarray}\label{multiIDoneEH_5_1}\nonumber
\tilde{\bf H}_{21}^{(1)} = \left[\begin{array}{c}{\bf H}_{21} \\ \vdots \\ {\bf H}_{K1}  \end{array}\right]\in \mathbb{C}^{(K-1)M\times M},~\tilde{\bf H}_{22} = diag \{{\bf H}_{22},{\bf H}_{33},...,{\bf H}_{KK} \}\in \mathbb{C}^{(K-1)M\times (K-1)M},
\end{eqnarray}
\begin{eqnarray}\label{multiIDoneEH_5_2}\nonumber
\tilde{\bf Q}_{2} =diag\{{\bf Q}_{2}, {\bf Q}_{3},..., {\bf Q}_{K} \}, ~\tilde{\bf R}_{-2}^{(1)} = {\bf I}_{(K-1)M} + {\bf C}_{I}+\tilde{\bf H}_{21}^{(1)}{\bf Q}_1(\tilde{\bf H}_{21}^{(1)})^H,
\end{eqnarray}
Here,
\begin{eqnarray}\label{multiIDoneEH_8_1_rev3}
{\bf C}_{I} = diag\{\sum_{j\neq 1,2}^{K}{\bf H}_{2j}{\bf Q}_{j}{\bf H}_{2j}^H ,\sum_{j\neq 1,3}^{K}{\bf H}_{3j}{\bf Q}_{j}{\bf H}_{3j}^H ,...,\sum_{j\neq 1,K}^{K}{\bf H}_{Kj}{\bf Q}_{j}{\bf H}_{Kj}^H  \} \succeq 0,
\end{eqnarray}
which is a block diagonal matrix. %That is, $C_k {\bf Q}_k$ with a constant scaling factor $C_k$ yields the boundary point of (\ref{oneIDmultiEH_2}).
\end{prop}
\begin{proof}
Let us consider the boundary point $(\bar R, \bar E)$ of (\ref{multiIDoneEH_1}) with any $\bar{\bf Q}_i$, $i=2,...,K$, the covariance matrices of the information transmitters on the boundary point. In addition, the corresponding covariance matrix of the energy transmitter at the boundary point is denoted as $\bar{\bf Q}_1$. From (\ref{multiIDoneEH_1}), the boundary of the harvested energy is defined as $g({\bf Q}_1)$ with respect to the covariance matrix of the first transmitter given as
\begin{eqnarray}\label{multiIDoneEH_rev2_1}
\!\!g({\bf Q}_1) =  tr ({\bf H}_{11} \bar{\bf
Q}_1 {\bf H}_{11}^H) +\sum_{\!j\!=\!2}^{\!K} tr ({\bf H}_{1j} \bar{\bf
Q}_j {\bf H}_{1j}^H) = tr (\tilde{\bf H}_{11}^{(1)} {\bf
Q}_1 (\tilde{\bf H}_{11}^{(1)})^H) + tr (\tilde{\bf H}_{12} {\bar{\tilde{\bf
Q}}}_2 \tilde{\bf H}_{12}^H) ,\!\!
\end{eqnarray}
where ${\bar{\tilde{\bf Q}}}_2 = diag\{\bar{\bf Q}_{2}, \bar{\bf Q}_{3},..., \bar{\bf Q}_{K} \}$ and $g({\bf Q}_1)$ is the exactly same form with the boundary of (\ref{multiIDoneEH_4}). Now we define functions $f({\bf Q}_1)$ with respect to ${\bf Q}_1$ as
\begin{eqnarray}\label{multiIDoneEH_8_1}
\!\!f({\bf Q}_1)\!\!\!&\!\!\!=\!\!\!&\!\!\sum_{i\!=\!2}^{K}\!\log \!\det({\bf I}_{M} \!+\! {\bf H}_{ii}^H\!{\bf R}_{-i}^{-1}{\bf H}_{ii}\bar{\bf Q}_i \!)\!,\\
 \!\!\!\!&\!\!\!=\!\!&\!\!  \log \!\det({\bf I}_{(K\!-\!1)M}\! +\! \tilde{\bf H}_{22}^H\!({\bf I}_{(K\!-\!1)M} \!+ \bar{\bf C}_{I}+\! {\bf S}({\bf Q}_1)\!)^{-\!1}\!\tilde{\bf
H}_{22}{\bar{\tilde{\bf Q}}}_2 ),\!\nonumber
\end{eqnarray}
where $\bar{\bf C}_{I} = diag\{\sum_{j\neq 1,2}^{K}{\bf H}_{2j}\bar{\bf Q}_{j}{\bf H}_{2j}^H ,\sum_{j\neq 1,3}^{K}{\bf H}_{3j}\bar{\bf Q}_{j}{\bf H}_{3j}^H ,...,\sum_{j\neq 1,K}^{K}{\bf H}_{Kj}\bar{\bf Q}_{j}{\bf H}_{Kj}^H  \}$ which is analogous to (\ref{multiIDoneEH_8_1_rev3}) and ${\bf S}({\bf Q}_1) =  diag\{ {\bf H}_{21}{\bf Q}_{1}{\bf H}_{21}^H ,...,  {\bf H}_{K1}{\bf Q}_{1}{\bf H}_{K1}^H \}$. Then, the boundary point $(\bar R, \bar E)$ of (\ref{multiIDoneEH_1}) at any boundary can be rewritten as $(\bar R, \bar E) = (f(\bar{\bf Q}_1), g(\bar{\bf Q}_1))$. Therefore, if $g(\bar{\bf Q}'_1) > g(\bar{\bf Q}_1)$ with $\bar{\bf Q}_1$ of the boundary point,
\begin{eqnarray}\label{multiIDoneEH_8_1_rev1}
f(\bar{\bf Q}'_1) < f(\bar{\bf Q}_1).
\end{eqnarray}
That is, the variation on $\bar{\bf Q}_1$ that increases $g(\bar{\bf Q}_1)$ always incurs a loss in $f(\bar{\bf Q}_1)$ at the boundary. From Lemma \ref{lem3}, (\ref{multiIDoneEH_8_1_rev1}) implies that
\begin{eqnarray}\label{multiIDoneEH_8_1_rev2}
\det({\bf I}_{(K\!-\!1)M} \! + \bar{\bf C}_I+\! {\bf S}(\bar{\bf Q}'_1)\!) >\det({\bf I}_{(K\!-\!1)M}+\bar{\bf C}_I \!+\! {\bf S}(\bar{\bf Q}_1)\!).
\end{eqnarray}
In addition, because the block diagonal entries of ${\bf S}(\bar{\bf Q}_1)$ correspond to those of $\tilde{\bf H}_{21}^{(1)}\bar{\bf Q}_1(\tilde{\bf H}_{21}^{(1)})^H$ and $\bar{\bf C}_{I}$ is also the block diagonal, from \cite{Ipsen} and Section 6.2 of \cite{RHorn2}, (\ref{multiIDoneEH_5_2}) implies that $\det({\bf I}_{(K-1)M} +\bar{\bf C}_{I} +\tilde{\bf H}_{21}^{(1)}\bar{\bf Q}'_1(\tilde{\bf H}_{21}^{(1)})^H) > \det({\bf I}_{(K-1)M}+\bar{\bf C}_{I} +\tilde{\bf H}_{21}^{(1)}\bar{\bf Q}_1(\tilde{\bf H}_{21}^{(1)})^H)$.  Accordingly, by letting
\begin{eqnarray}\label{multiIDoneEH_8_4}
\!\!&f'({\bf Q}_1)\! =\!\log\! \det({\bf I}_{(K-1)M} + \tilde{\bf H}_{22}^H({\bf I}_{(K-1)M}+\bar{\bf C}_{I} +\tilde{\bf H}_{21}^{(1)}{\bf Q}_1(\tilde{\bf H}_{21}^{(1)})^H)^{-1}\tilde{\bf
H}_{22}{\bar{\tilde{\bf Q}}}_2 ),&\nonumber
\end{eqnarray}
it can be found that, if $f(\bar{\bf Q}'_1) < f(\bar{\bf Q}_1)$,
\begin{eqnarray}\label{multiIDoneEH_8_4_rev1}
f'(\bar{\bf Q}'_1) < f'(\bar{\bf Q}_1).
\end{eqnarray}
That is, the variation on $\bar{\bf Q}_1$ that increases $g(\bar{\bf Q}_k)$ always incurs a loss in $f'(\bar{\bf Q}_1)$.
Therefore, all the optimal $\bar{\bf Q}_1$ also becomes the solutions for the boundary of (\ref{multiIDoneEH_4}).
Accordingly, all the optimal $\bar{\bf Q}_1$ for $(\bar R, \bar E)$ of (\ref{multiIDoneEH_1}) become the solutions for the boundary of
(\ref{multiIDoneEH_4}).
\end{proof}
\begin{remark}\label{remark_multiIDoneEH}
Note that (\ref{multiIDoneEH_4}) can be regarded as the achievable R-E tradeoff region for two-user MIMO
IFC with an effective channel set of ($\tilde{\bf H}_{11}^{(1)}$, $\tilde{\bf H}_{12}$, $\tilde{\bf H}_{21}^{(1)}$, $\tilde{\bf H}_{22}$), the additional constraint on the covariance matrix of $\tilde{\bf x}_2$, and the covariance matrix of the interference ${\bf C}_I$ that is unknown to the energy transmitter as in Fig. \ref{figK_user_OneEH} (a). Because the rank-one optimality of the effective two-user MIMO IFC with the interference unknown to the energy transmitter in Corollary \ref{cor_oneIDmultiEH} is still valid even with the block-diagonal structure on the covariance matrix at the information transmitter, i.e. $\tilde{\bf Q}_2$, we can have the following corollary.
\begin{cor}\label{cor_multiIDoneEH} (Multiple IDs and one EH) The optimal ${\bf Q}_1$ at the boundary
of the achievable rate-energy region (\ref{multiIDoneEH_1}) has a rank one at most. That
is, $rank ({\bf Q}_1) \leq 1$.
\end{cor}
\end{remark}
\begin{figure}%[htbp]
\centering %\hspace{-3em}
 \subfigure[]
  {\includegraphics[height=3.3cm]{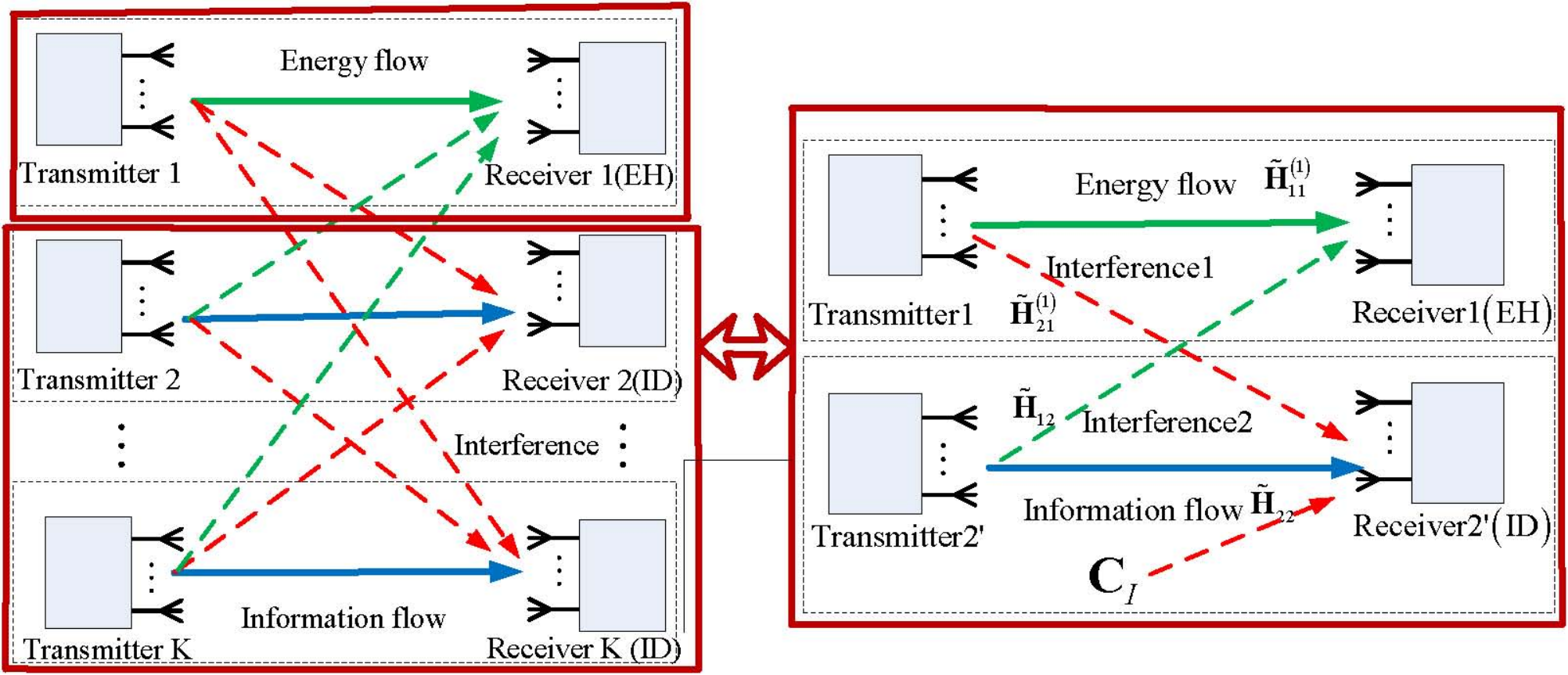}}
 \subfigure[]
  {\includegraphics[height=4.0cm]{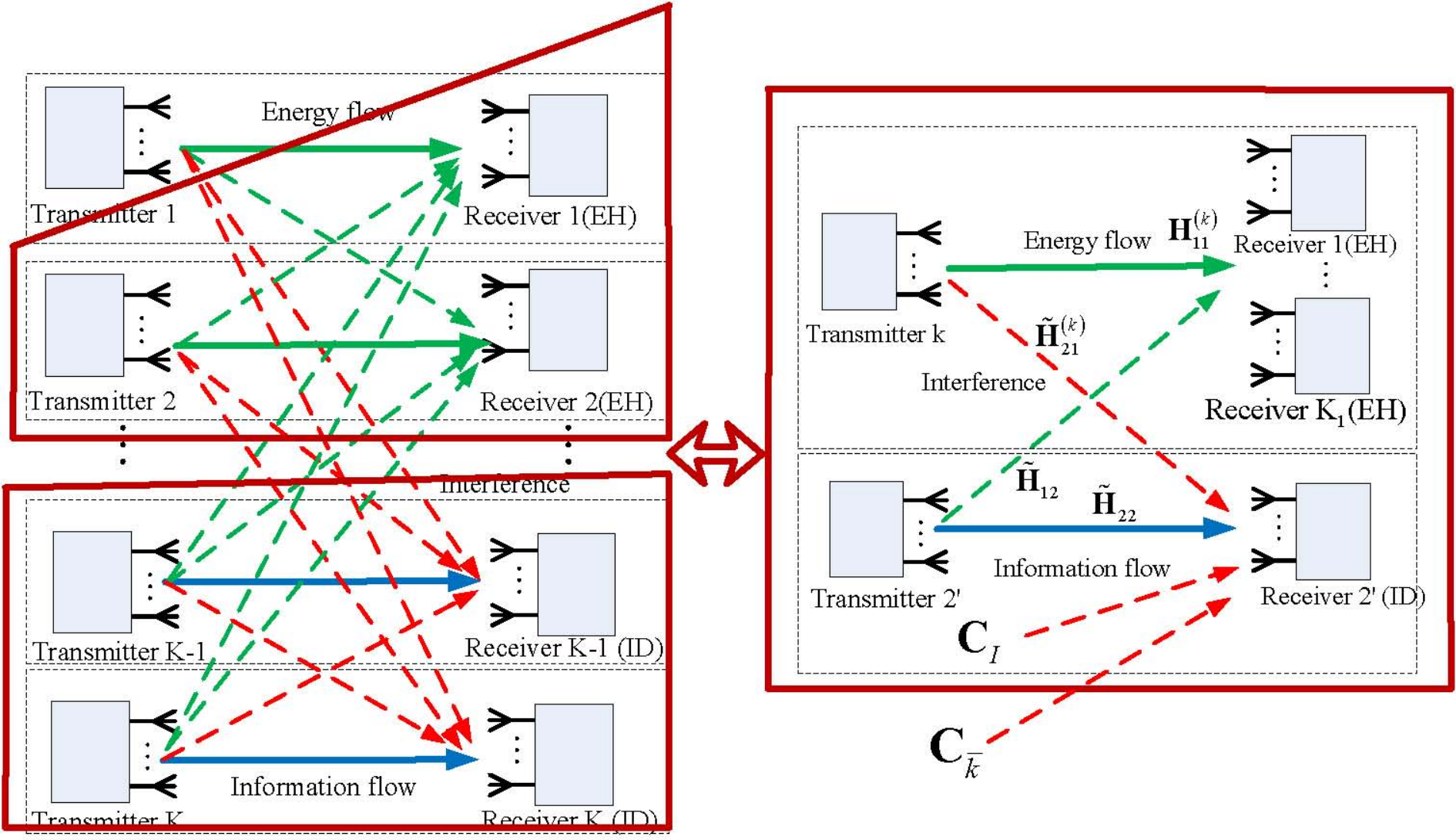}}
 \caption{(a) One EH and multiple IDs and (b) Multiple EHs and multiple IDs in a K-user IFC.} \label{figK_user_OneEH}
\end{figure}

\subsection{multiple EHs and multiple IDs in a K-user IFC}\label{sec:multiplemultiple}

Now let us consider that multiple EHs and multiple IDs coexist. That is, the transceiver pair $(Tx_i, Rx_i)$ for $i=1,...,K_1$ operate in EH mode, while $(Tx_i, Rx_i)$, $k=K_1+1,...,K$ in ID mode. By letting $R=\sum_{i=K_1+1}^{K} R_k$ and $E =\sum_{i=1}^{K_1} E_i$ with $E_i = \sum_{j=1}^{K} E_{ij}$, we can define
the achievable rate-energy region as:
\begin{eqnarray}\label{multiIDmultiEH_1}
\!\!&\!C_{R\!-\!E} (P) \!\triangleq  \!\Biggl\{ \!(R, E) : R \leq
\sum_{i=K_1+1}^{K}\log \det({\bf I}_{M} + {\bf H}_{ii}^H{\bf R}_{-i}^{-1}{\bf
H}_{ii}{\bf Q}_i ),\!&\!\nonumber\\
\!\!&\!E \!\leq \!  \sum_{\!i\!=\!1}^{\!K_1}\sum_{\!j\!=\!1}^{\!K} tr ({\bf H}_{ij} {\bf
Q}_j {\bf H}_{ij}^H), tr({\bf Q}_j)\!\leq \!P, {\bf Q}_j\!\succeq
\!{\bf 0}, j\!=\!1,...,\!K\! \Biggr\}\!.\!&\!
\end{eqnarray}
Then we have the following proposition.
\begin{prop}\label{prop_multiIDmultiEH} All the optimal ${\bf Q}_k$ at the boundary
of the achievable rate-energy region for (\ref{multiIDmultiEH_1}) also become the optimal solutions for the boundary of
\begin{eqnarray}\label{multiIDmultiEH_2}
\!\!&\!C_{R\!-\!E,k} (P) \!\triangleq  \!\Biggl\{ \!(R, E) : R \leq
\log \det({\bf I}_{(K-K_1)M} + \tilde{\bf H}_{22}^H(\tilde{\bf R}_{-2}^{(k)})^{-1}\tilde{\bf
H}_{22}\tilde{\bf Q}_2 ),\!&\!\nonumber\\
\!\!&\!E \!\leq \! tr (\tilde{\bf H}_{11}^{(k)} {\bf
Q}_{k} (\tilde{\bf H}_{11}^{(k)})^H) + tr (\tilde{\bf H}_{12} \tilde{\bf
Q}_2 \tilde{\bf H}_{12}^H) , tr({\bf Q}_j)\!\leq \!P, {\bf Q}_j\!\succeq
\!{\bf 0},  j\!=\!k,K_1+1,...,K\!\Biggr\}\!,\!&\!
\end{eqnarray}
for all $k=1,...,K_1$,
where
\begin{eqnarray}\label{multiIDmultiEH_4}
\tilde{\bf H}_{11}^{(k)} = \left[\begin{array}{c}{\bf H}_{1k} \\ \vdots \\ {\bf H}_{K_1k}  \end{array}\right], ~~\tilde{\bf H}_{12} =\left[\tilde{\bf H}_{12}^{(K_1+1)},...,\tilde{\bf H}_{12}^{(K)}\right] = \left[\begin{array}{ccc}{\bf H}_{1K_1+1} & \hdots & {\bf H}_{1K} \\ \vdots & \ddots & \vdots\\ {\bf H}_{K_1K_1+1} &\hdots &{\bf H}_{K_1K}  \end{array}\right]\in\mathbb{C}^{K_1M\times(K-K_1)M},
\end{eqnarray}
\begin{eqnarray}\label{multiIDmultiEH_5}
\tilde{\bf H}_{21}^{(k)} = \left[\begin{array}{c}{\bf H}_{K_1+1k} \\ \vdots \\ {\bf H}_{Kk}  \end{array}\right], ~~\tilde{\bf H}_{22} =diag\{{\bf H}_{K_1+1K_1+1}, {\bf H}_{K_1+2K_1+2},...,{\bf H}_{KK} \}\in\mathbb{C}^{(K-K_1)M\times(K-K_1)M},
\end{eqnarray}
and
$ \tilde{\bf Q}_{2} = diag\{{\bf Q}_{K_1+1},...,{\bf Q}_{K} \}, ~~\tilde{\bf R}_{-2}^{(k)} = {\bf I}_{(K-K_1)M} + {\bf C}_I +{\bf C}_{\bar k}+ \tilde{\bf H}_{21}^{(k)}{\bf Q}_k(\tilde{\bf H}_{21}^{(k)})^H \in\mathbb{C}^{(K-K_1)M\times(K-K_1)M}$. Here, ${\bf C}_{\bar k} =  \sum_{\!j\!=\!1, j\!\neq\! k}^{\!K_1}   \tilde{\bf H}_{21}^{(j)} {\bf
Q}_j (\tilde{\bf H}_{21}^{(j)})^H$ and $\displaystyle{{\bf C}_{I} = diag\left\{\sum_{\substack{j=K_1+1,\\ j\neq K_1+1}}^{K}{\bf H}_{K_1+1j}{\bf Q}_{j}{\bf H}_{K_1+1j}^H ,...,\sum_{\substack{j=K_1+1,\\ j\neq K}}^{K}{\bf H}_{Kj}{\bf Q}_{j}{\bf H}_{Kj}^H  \right\}}$.
%That is, $C_k {\bf Q}_k$ with a constant scaling factor $C_k$ yields the boundary point of (\ref{oneIDmultiEH_2}).
\end{prop}
\begin{proof}
Let us consider the boundary point $(\bar R, \bar E)$ of the achievable rate-energy for (\ref{multiIDmultiEH_1}) for any given $\bar{\bf Q}_{K_1+1},...,\bar{\bf Q}_{K}$. In addition, the corresponding covariance matrix of the $k$th energy transmitter at the boundary point are denoted as $\bar{\bf Q}_k$. Then, from Proposition \ref{prop_oneIDmultiEH}, $\bar{\bf Q}_k$ for all $(\bar R, \bar E)$ of (\ref{multiIDmultiEH_1}) also become the solutions for the boundary of
\begin{eqnarray}\label{multiIDmultiEH_8}
\!\!&\!C_{R\!-\!E,k}' (P) \!\triangleq  \!\Biggl\{ \!(R, E) : R \leq
\sum_{i=K_1+1}^{K}\log \det({\bf I}_{M}+ {\bf H}_{ii}^H{\bf R}_i^{-1}{\bf
H}_{ii}{\bf Q}_i ),\!&\!\nonumber\\
\!\!&\!E \!\leq \!\sum_{\!i\!=\!1}^{\!K_1} (tr({\bf H}_{ik} {\bf
Q}_k {\bf H}_{ik}^H) + \sum_{\!j\!=\!K_1+1}^{\!K} tr ({\bf H}_{ij} {\bf
Q}_j {\bf H}_{ij}^H) ), tr({\bf Q}_j)\!\leq \!P, {\bf Q}_j\!\succeq
\!{\bf 0}, j\!=\!1,...,\!K\! \Biggr\}\!\!&\!\nonumber\\
&\!= \!\Biggl\{ \!(R, E) : R \leq
\sum_{i=K_1+1}^{K}\log \det({\bf I}_{M}+ {\bf H}_{ii}^H{\bf R}_i^{-1}{\bf
H}_{ii}{\bf Q}_i ),\!&\!\nonumber\\
\!\!&\!E \!\leq \!tr(\tilde{\bf H}_{11}^{(k)}{\bf
Q}_k (\tilde{\bf H}_{11}^{(k)})^H)+\sum_{\!j\!=\!K_1+1}^{\!K} tr ( \tilde{\bf H}_{12}^{(j)} {\bf
Q}_j (\tilde{\bf H}_{12}^{(j)})^H)  , tr({\bf Q}_j)\!\leq \!P, {\bf Q}_j\!\succeq
\!{\bf 0}, j\!=\!k,K_1+1,...,\!K\! \Biggr\}\!.\!&\!
\end{eqnarray}
In (\ref{multiIDmultiEH_8}),
\begin{eqnarray}\label{multiIDmultiEH_7}
&\log \det({\bf I}_{M} + {\bf H}_{ii}^H{\bf R}_{-i}^{-1}{\bf
H}_{ii}{\bf Q}_i ) =  &\nonumber\\ &\log \det({\bf I}_{M}+ {\bf H}_{ii}^H({\bf I}_M + \sum_{\substack{ j=1 \\
j\neq k}}^{K_1}{\bf H}_{ij}{\bf Q}_j{\bf H}_{ij}^H+ \sum_{\substack{ j=K_1+1 \\
j\neq i}}^{K}{\bf H}_{ij}{\bf Q}_j{\bf H}_{ij}^H+{\bf H}_{ik}{\bf Q}_k{\bf H}_{ik}^H)^{-1}{\bf
H}_{ii}{\bf Q}_i ).&
\end{eqnarray}
Here, $\sum_{\substack{ j=1 \\
j\neq k}}^{K_1}{\bf H}_{ij}{\bf Q}_j{\bf H}_{ij}^H$ and $\sum_{\substack{ j=K_1+1 \\
j\neq i}}^{K}{\bf H}_{ij}{\bf Q}_j{\bf H}_{ij}^H$ are the interferences from the energy-transferring transmitters except the $k$th transmitter and from other information-transferring transmitters, respectively. In addition, they are equal to the $i$th $M\times M$ diagonal block entry of ${\bf C}_{\bar k}$ and ${\bf C}_{I}$, respectively. Therefore, the R-E region of (\ref{multiIDmultiEH_8}) is equivalent with the case of one EH and $K-K_1$ IDs in (\ref{multiIDoneEH_1}) where the interference whose covariance matrix is given as the $i$th $M\times M$ diagonal block entry of ${\bf C}_{\bar k}$ is added to the $i$th ID receiver. Therefore, similarly to what is done in the proof of Proposition \ref{prop_multiIDoneEH}, $\bar{\bf Q}_k$ for all $(\bar R, \bar E)$ of (\ref{multiIDmultiEH_8}) also become the solutions for the boundary of (\ref{multiIDmultiEH_2}). This implies that the optimal $\bar{\bf Q}_k$, $k=1,...,K_1$ for $(\bar R, \bar E)$ of (\ref{multiIDmultiEH_1}) also yields the boundary of (\ref{multiIDmultiEH_2}). See also Fig. \ref{figK_user_OneEH} (b).
\end{proof}
\begin{cor}\label{cor_multiIDmultiEH} (Multiple IDs and multiple EHs) The optimal ${\bf Q}_k$, $k=1,...,K_1$ at the boundary
of the achievable rate-energy region (\ref{multiIDmultiEH_1}) has a rank one at most. That
is, $rank ({\bf Q}_k) \leq 1$ for $k=1,...,K_1$.
\end{cor}
\begin{remark}\label{remark_multiIDmultiEH}
From Proposition \ref{prop_multiIDmultiEH}, when transferring the energy in K-user MIMO IFC, the energy transmitters' optimal strategy is a rank-one beamforming with a proper power allocation, which is a generalized version of Proposition \ref{prop1} for two-user MIMO IFC \cite{ParkBruno}. That is, if the covariance matrix of an energy transmitter in the general K-user MIMO IFC has a rank ($\geq 2$), a rank-one beamforming for that transmitter exhibiting either higher information rate or larger harvested energy can be found. However, as observed in (\ref{oneIDmultiEH_2rev4}) of Corollary \ref{cor_oneIDmultiEH}, the optimal beamforming depends on the covariance matrix of the interference from other energy/information transmitters (specifically, the beamforming directions of the other energy transmitters and the covariance matrices of information transmitters). For example, from (\ref{oneIDmultiEH_2rev4}) and $\tilde{\bf R}_{-2}^{(k)} $ of (\ref{multiIDmultiEH_5}), the beamforming vector minimizing the interference, $\bar{\bf v}_I$, is given as the right singular vector associated with the smallest singular value of
\begin{eqnarray}\label{multiIDmultiEH_7_rev2}
(\tilde{\bf H}_{21}^{(k)})^H{\bf U}_{\tilde C}({\bf I}_{(K-K_1)M} +  {\bf \Sigma}_{\tilde C})^{-1}{\bf U}_{\tilde C}^H\tilde{\bf H}_{21}^{(k)},
\end{eqnarray}
where $\tilde{\bf C}_{\bar k} \triangleq {\bf C}_{\bar k} + {\bf C}_{I} = {\bf U}_{\tilde C}{\bf \Sigma}_{\tilde C} {\bf U}_{\tilde C}^H$. Note that when $\tilde{\bf C}_{\bar k}= {\bf 0}$, $\bar{\bf v}_I$ becomes the singular vector associated with the smallest singular value of the effective channel $\tilde{\bf H}_{21}^{(k)}$ in two-user MIMO IFC. This is important because, if the transmitters know the global CSI and the {\it{centralized}} optimization is possible, they can align the energy beams properly in cooperation with other transmitters (i.e., beam alignment) and we can further improve the R-E region, which is out of scope of this paper. Instead, in the next section, we will present the distributed rank-one beamforming strategies and, in the simulation, we show that the system performance (achievable rate and harvested energy) can be further improved by using beam tilting strategy jointly with the distributed rank-one beamforming.
%Furthermore, given a rank-one beamforming, the boundary region (\ref{multiIDmultiEH_2}) can be an upper-bound of the original R-E region. Even %though it is not tight (because the interferences from all other information transmitters are removed), it gives an insight on the design of %iterative algorithm for (\ref{multiIDmultiEH_1}).
\end{remark}

\section{Distributed Rank-one Beamforming design and Achievable R-E region}\label{sec:rankoneBF_REregion}

In this section, we propose distributed rank-one beamforming methods based on Proposition \ref{prop_multiIDmultiEH}, and then propose an iterative algorithm that computes the achievable R-E trade-off curves for the K-user MIMO IFC with different beamforming schemes. Again, let us consider that the transceiver pair $(Tx_i, Rx_i)$ for $i=1,...,K_1$ operate in EH mode, while $(Tx_i, Rx_i)$, $k=K_1+1,...,K$ operate in ID mode, without loss of generality.

\subsection{Distributed Rank-one Beamforming Design}\label{ssec:rankoneBF}

Because there exists multiple EH receivers, each energy transferring transmitter steers its signal to maximize the energy transferred to all EH receivers. Therefore, because ${\bf Q}_k$ for $k=1,...,K_1$ has a rank one from Proposition \ref{prop_multiIDmultiEH}, it can be given by
\begin{eqnarray}\label{MEB}
{\bf Q}_k = P_k [\tilde{\bf V}_{11}^{(k)}]_1[\tilde{\bf V}_{11}^{(k)}]_1^H,
\end{eqnarray}
where $\tilde{\bf V}_{11}^{(k)}$ is a $M \times M$ unitary matrix obtained
from the SVD of $\tilde{\bf H}_{11}^{(k)}$ and $0\leq P_k \leq P$. That is,
$\tilde{\bf H}_{11}^{(k)} = \tilde{\bf U}_{11}^{(k)} \tilde{\bf \Sigma}_{11}^{(k)} (\tilde{\bf V}_{11}^{(k)})^H$,
where $\tilde{\bf \Sigma}_{11}^{(k)} = diag\{\sigma_{1}(\tilde{\bf H}_{11}^{(k)}),..., \sigma_{
 M}(\tilde{\bf H}_{11}^{(k)}) \}$. That is, the $k$th transmitter's beamforming is analogous to the maximum-energy beamforming (MEB) on the two-user MIMO IFC (as in \cite{ParkBruno}) but applied to the effective channel $\tilde{\bf H}_{11}^{(k)}$. Here,
the energy harvested from the $k$th transmitter is given by $P_k
(\sigma_{1}(\tilde{\bf H}_{11}^{(k)}))^2$.

From an ID perspective
at the ID receivers, the $k$th transmitter should steer its
signal to minimize the interference power to all the ID receivers.
That is, the corresponding transmit
covariance matrix ${\bf Q}_k$ is then given by
\begin{eqnarray}\label{MLB}
{\bf Q}_k = P_k [\tilde{\bf V}_{21}^{(k)}]_M[\tilde{\bf V}_{21}^{(k)}]_M^H,
\end{eqnarray}
where $\tilde{\bf V}_{21}^{(k)}$ is a $M \times M$ unitary matrix obtained
from the SVD of $\tilde{\bf H}_{21}^{(k)}$ and $0\leq P_k\leq P$. That is,
$\tilde{\bf H}_{21}^{(k)} = \tilde{\bf U}_{21}^{(k)} \tilde{\bf \Sigma}_{21}^{(k)} (\tilde{\bf V}_{21}^{(k)})^H$,
where $\tilde{\bf \Sigma}_{21}^{(k)} = diag\{\sigma_{1 }(\tilde{\bf H}_{21}^{(k)}),..., \sigma_{
M }(\tilde{\bf H}_{21}^{(k)}) \}$. The $k$th transmitter's beamforming that minimizes the interference to the effective channel $\tilde{\bf H}_{21}^{(k)}$ is also analogous to the {\it{minimum-leakage beamforming (MLB)}} in the two-user MIMO IFC \cite{ParkBruno}.
Then, the energy harvested from the $k$th transmitter is given by $P_k\|\tilde{\bf H}_{11,k} [\tilde{\bf V}_{21,k}]_M\|^2$.

Because MEB and MLB strategies are developed according to different aims - either
maximizing transferred energy to EH or minimizing interference
(or, leakage) to ID, respectively, they have their own weaknesses - causing either large interference to ID receivers or harvesting insufficient energy at the EH receivers.

\subsubsection{Energy-regularized SLER-maximizing beamforming}
\label{ssec:SLER_maximizing}

To maximize the transferred energy to EH and
simultaneously minimize the leakage to ID, we define a new
performance metric, signal-to-leakage-and-harvested energy ratio
(SLER) at the $k$th transmitter as \cite{ParkBruno}
\begin{eqnarray}\label{GSVD1}
SLER_k &=&\frac{\|\tilde{\bf H}_{11}^{(k)}{\bf v}_k \|^2}{\|\tilde{\bf H}_{21}^{(k)}{\bf v}_k \|^2
+ max(\bar E/K_1 -P\|\tilde{\bf H}_{11}^{(k)}\|^2 ,0)} \nonumber\\
&=&\frac{{\bf v}_k^H(\tilde{\bf H}_{11}^{(k)})^H\tilde{\bf H}_{11}^{(k)}{\bf v}_k }{{\bf
v}_k^H\left((\tilde{\bf H}_{21}^{(k)})^H\tilde{\bf H}_{21}^{(k)} +{max({\bar E}/{K_1 P} -\|\tilde{\bf
H}_{11}^{(k)}\|^2 ,0)}{\bf I}_M \right){\bf v}_k }.
\end{eqnarray}
The beamforming vector ${\bf v}_k$ that maximizes SLER of
(\ref{GSVD1}) is then given by
\begin{eqnarray}\label{GSVD3}
{\bf v}_k = \sqrt{P_k}\frac{ \tilde{\bf v}_k}{\|\tilde{\bf v}_k\|},
\end{eqnarray}
where $\tilde{\bf v}_k$ is the generalized eigenvector associated with
the largest generalized eigenvalue of the matrix pair
\begin{eqnarray}\label{GSVD4}
((\tilde{\bf
H}_{11}^{(k)})^H\tilde{\bf H}_{11}^{(k)}, (\tilde{\bf H}_{21}^{(k)})^H\tilde{\bf H}_{21}^{(k)} +{max(\bar E
/{K_1 P} -\|\tilde{\bf H}_{11}^{(k)}\|^2 ,0)}{\bf I}_M).
\end{eqnarray}
Here, $\tilde{\bf v}_k$
can be efficiently computed by using a GSVD algorithm \cite{JPark, JPark2}.

\begin{remark}\label{remark_SLERBF1} Note that the SLER metric is comparable with the signal-to-leakage-and-noise (SLNR) ratio \cite{JPark} which is widely utilized in the precoding design for the information transfer in the multi-user MIMO system. That is, the noise power contributes to the denominator of SLNR in the beamforming design \cite{JPark} because the noise at the receiver together with the leakage to other receivers affects the system performance degradation for the information transfer. In contrast, the contribution of the minimum required harvested energy is added in SLER, because the required harvested energy minus the energy directly harvested from the $k$th transmitter is the main performance barrier of the EH receiver. Interestingly, from (\ref{GSVD1}), when the required harvested
energy at the EH receiver is large, the matrix $(\!\tilde{\bf H}_{21}^{\!(\!k\!)\!}\!)^H\!\tilde{\bf H}_{21}^{\!(\!k\!)\!}\! +\!{max({\bar E}/{K_1 \!P} \!-\!\|\!\tilde{\bf
H}_{11}^{\!(\!k\!)\!}\!\|^2 ,0)}{\bf I}_M $ in the denominator
of (\ref{GSVD1}) approaches an identity matrix multiplied by a
scalar. Accordingly, the SLER maximizing beamforming is equivalent
with the MEB in (\ref{MEB}). That is, ${\bf v}_k$
becomes $\sqrt{P_k}[\tilde{\bf V}_{11}^{(k)}]_1$. In contrast, as the required
harvested energy becomes smaller, ${\bf v}$ is steered such that
less interference is leaked into the ID receiver to reduce the
denominator of (\ref{GSVD1}). That is, ${\bf v}$ approaches the
MLB weight vector in (\ref{MLB}). Therefore, the
proposed SLER maximizing beamforming balances both metrics -
energy maximization to EH and leakage minimization to ID, which has been confirmed in \cite{ParkBruno}.
\end{remark}

\subsection{Achievable R-E region}\label{ssec:REregion}

Note that the achievable sum-rate is unknown for the general K-user MIMO IFC and, accordingly, the optimal region for (\ref{multiIDmultiEH_2}) is not easily identified. Instead, motivated by Corollary \ref{cor_multiIDmultiEH} that each $k$th transmitter for $k=1,...,K_1$ transfers its signal with a rank-one beamforming at most, we propose an iterative algorithm that optimizes the transmit powers $P_k$, $k=1,...,K_1$ and ${\bf Q}_k$, $k=K_1+1,...,K$ simultaneously. The energy transmitters can choose their covariance matrices among (\ref{MEB}), (\ref{MLB}), (\ref{GSVD3}), or other rank-one covariance matrices. In this paper, we assume that they adopt the same beamforming strategy among MLB, MEB, and SLER beamforming and compare their performance by simulation.
Given ${\bf Q}_k$ as in (\ref{MEB}), (\ref{MLB}), or (\ref{GSVD3}),
the achievable rate-energy region is then given as:
\begin{eqnarray}\label{REregion_1}
\!\!&\!C_{R\!-\!E} (P) \!=  \!\Biggl\{ \!(R, E) : R \leq
\sum_{i=K_1+1}^{K}\log \det({\bf I}_{M} + {\bf H}_{ii}^H{\bf R}_{-i}^{-1}{\bf
H}_{ii}{\bf Q}_i ),\!&\!\nonumber\\
\!\!&\!E \!\leq \!\sum_{\!j\!=\!K_1+1}^{\!K} tr ( \tilde{\bf H}_{12}^{(j)} {\bf
Q}_j (\tilde{\bf H}_{12}^{(j)})^H) + E_{11}, tr({\bf Q}_j)\!\leq \!P, {\bf Q}_j\!\succeq
\!{\bf 0}, j\!=\!1,...,\!K\! \Biggr\}\!.\!&\!
\end{eqnarray}
where
\begin{eqnarray}\label{REregion_2}
E_{11} = \sum_{\!j\!=\!1}^{\!K_1}tr(\tilde{\bf H}_{11}^{(j)}{\bf
Q}_k (\tilde{\bf H}_{11}^{(j)})^H)= \sum_{\!j\!=\!1}^{\!K_1} \omega_j P_j,
\end{eqnarray}
with
\begin{eqnarray}\label{REregion_3}
\omega_j  = \left\{\begin{array}{c}\|\tilde{\bf H}_{11}^{(j)} [\tilde{\bf V}_{11}^{(j)}]_1 \|^2  \quad{\text{for MEB}}\\
\|\tilde{\bf H}_{11}^{(j)} [\tilde{\bf V}_{21}^{(j)}]_M\|^2 \quad {\text{for MLB}} \\ \|\tilde{\bf H}_{11}^{(j)}\frac{ \tilde{\bf v}_j}{\|\tilde{\bf v}_j\|}\|^2 \quad {\text{for SLER beamforming}} \end{array} \right. ,
\end{eqnarray}
and
\begin{eqnarray}\label{REregion_4}
{\bf R}_{-i} = {\bf I}_{M}+ \sum_{ j=1}^{K_1}P_j{\bf \Omega}_{ij}+ \sum_{\substack{ j=K_1+1 \\
j\neq i}}^{K}{\bf H}_{ij}{\bf Q}_j{\bf H}_{ij}^H,
\end{eqnarray}
with
\begin{eqnarray}\label{REregion_5}
{\bf \Omega}_{ij} = \left\{\begin{array}{c}{\bf H}_{ij} [\tilde{\bf V}_{11}^{(j)}]_1 [\tilde{\bf V}_{11}^{(j)}]_1^H{\bf H}_{ij}^H  \quad{\text{for MEB}}\\
{\bf H}_{ij}[\tilde{\bf V}_{21}^{(j)}]_M[\tilde{\bf V}_{21}^{(j)}]_M^H{\bf H}_{ij}^H\quad {\text{for MLB}} \\{\bf H}_{ij} \tilde{\bf v}_j \tilde{\bf v}_j^H{\bf H}_{ij}^H /\|\tilde{\bf v}_j\|^2 \quad {\text{for SLER beamforming}} \end{array} \right. .
\end{eqnarray}
Accordingly, we have the following optimization problem for the rate-energy region of (\ref{REregion_1})
%\footnote{The dual problem of maximizing
%energy subject to rate constraint can be formulated, but the rate
%maximization problem is preferred because it can be solved using
%approaches similar to those in the rate maximization problems
%under various constraints \cite{Scutari, Zhang1, }.}
\begin{eqnarray}\label{REregion_6}
\!\!(\!P1\!)\! \underset{\substack{{\bf Q}_{k},k=K_1+1,...,K\\P_i,i=1,...,K_1}}{\text{ maximize}}& J\triangleq \sum_{i=K_1+1}^{K}\log \det({\bf I}_{M} + {\bf H}_{ii}^H{\bf R}_{-i}^{-1}{\bf
H}_{ii}{\bf Q}_i )\\\label{REregion_7} \!\!{\text{subject to}}\!&\!\!\sum_{\!j\!=\!K_1+1}^{\!K} tr ( \tilde{\bf H}_{12}^{(j)} {\bf
Q}_j (\tilde{\bf H}_{12}^{(j)})^H) \geq \max(\bar E \!-\!
E_{11},0)\!\\\label{REregion_8} &tr({\bf Q}_i) \leq P,~{\bf
Q}_i \succeq{\bf 0}, i=K_1,...,K\\
& P_j\leq P, j = 1,...,K_1,
\end{eqnarray}
where $\bar E$ can take any value less than $E_{\max}$ denoting
the maximum energy transferred from all the transmitters. Here, it
can be easily derived that $E_{\max}$ is given as
\begin{eqnarray}\label{REregion_9}
E_{\max} =P\left(  \sum_{\!j\!=\!1}^{\!K_1} \omega_j +\sum_{j=K_1+1}^{K} \sigma_1(\tilde{\bf H}_{12}^{(j)})\right),
\end{eqnarray}
which is obtained when all the information transmitters steer their signals such that their transferred energy is maximized on the cross link $\tilde{\bf H}_{12}^{(j)}$.
The optimization problem (P1) is obviously non-convex due to the coupled variables in the objective function $J$. That is, because of
the interference at each ID receiver from other information transmitters, ${\bf Q}_k$ are coupled in the objective function. Accordingly, here we develop a sub-optimal iterative algorithm for (P1).

Before we proceed with (P1), let us consider a simplified optimization problem by removing the interferences from other information transmitters by assuming that the cross-channel gain among the information transceivers is very small as
\begin{eqnarray}\label{REregion_10}\!\!(\!P1-UP\!)\! \underset{\substack{{\bf Q}_{k},k=K_1+1,...,K\\P_i,i=1,...,K_1}}{\text{ maximize}}& J^{UP}\triangleq\sum_{i=K_1+1}^{K}\log \det({\bf I}_{M} + {\bf H}_{ii}^H\bar{\bf R}_{-i}^{-1}{\bf
H}_{ii}{\bf Q}_i )\\\label{REregion_11} \!\!{\text{subject to}}\!&\!\!\sum_{\!j\!=\!K_1+1}^{\!K} tr ( \tilde{\bf H}_{12}^{(j)} {\bf
Q}_j (\tilde{\bf H}_{12}^{(j)})^H) \geq \max(\bar E \!-\!
E_{11},0)\!\\\label{REregion_12} &tr({\bf Q}_i) \leq P,~{\bf
Q}_i \succeq{\bf 0}, i=K_1,...,K\\\label{REregion_12_1}
& P_j\leq P, j = 1,...,K_1,
\end{eqnarray}
where $\bar{\bf R}_{-i} = {\bf I}_{M}+ \sum_{ j=1}^{K_1}P_j{\bf \Omega}_{ij}$. Note that, because the interferences from all other information transmitters are removed, (P1-UP) can be an upper-bound of the original R-E region. Even though it is not tight, it gives an insight on how to develop the iterative algorithm for the original problem.

By letting ${\bf P}=[P_1,...,P_{K_1}]^T$, using that $\frac{d}{dx}\log \det({\bf A}(x)) = tr\left({\bf A}(x)^{-1}\frac{d {\bf A}(x)}{dx}\right)$ \cite{Magnus} and Sylvester's determinant theorem \cite{DHarville}, $ \nabla_{{\bf P}} J^{UP}({\bf P}, {\bf Q}_{K_1+1},...,{\bf Q}_K) \in \mathbb{R}^{K_1 \times 1}$ is given as
\begin{eqnarray}\label{REregion_13}
\nabla_{{\bf P}} J^{UP}({\bf P}, {\bf Q}_{K_1+1},...,{\bf Q}_K) = \left[\begin{array}{c} tr\left(\sum_{i=K_1+1}^{K}\left(( {\bf H}_{ii}{\bf Q}_i{\bf H}_{ii}^H +\bar{\bf R}_{-i})^{-1} - ( \bar{\bf R}_{-i})^{-1}\right){\bf \Omega}_{i1} \right) \\\vdots\\tr\left(\sum_{i=K_1+1}^{K}\left(( {\bf H}_{ii}{\bf Q}_i{\bf H}_{ii}^H +\bar{\bf R}_{-i})^{-1} - ( \bar{\bf R}_{-i})^{-1}\right){\bf \Omega}_{iK_1} \right) \end{array}\right].
\end{eqnarray}
Because, from (\ref{REregion_13}), the objective function in (P1-UP) is monotonically decreasing with respect to ${\bf P}$, we iteratively optimize their values using the steepest descent method as:

\vspace*{5pt}Algorithm 1. {\it{\underline{Identification of the
achievable R-E region for P1-UP:}}}
\begin{enumerate}
\item Initialize $n=0$, ${\bf P}^{[0]}=[P,...,P]^T$,
\begin{eqnarray}\label{eqnRevise_Algo_2_1}
E_{11}^{[0]} = \sum_{\!j\!=\!1}^{\!K_1} \omega_j ({\bf P}^{[0]})_j,\quad \bar{\bf R}_{-i}^{[0]} = {\bf I}_{M}+ \sum_{ j=1}^{K_1}({\bf P}^{[0]})_j{\bf \Omega}_{ij}.
\end{eqnarray}
\item For $n=0:N_{max}$
\begin{enumerate}
\item Solve the optimization problem (P1-UP) for ${\bf Q}_{K_1+1}^{[n]},...,{\bf Q}_K^{[n]}$ as
a function of
$E_{11}^{[n]}$ and $\bar{\bf R}_{-i}^{[n]}$.
\item If $\sum_{\!j\!=\!K_1+1}^{\!K} tr ( \tilde{\bf H}_{12}^{(j)} {\bf
Q}_j^{[n]} (\tilde{\bf H}_{12}^{(j)})^H) +
E_{11}^{[n]}
>\bar E$
\begin{eqnarray}\label{eqnRevise_Algo_2_3} {\bf P}^{[n+1]} =max\left({\bf P}^{[n]} +\Delta\cdot \nabla_{{\bf P}} J^{UP}({\bf P}^{[n]}, {\bf Q}_{K_1+1}^{[n]},...,{\bf Q}_K^{[n]}), {\bf 0}\right),
\end{eqnarray}
where the step size $\Delta$ is given by a fixed value on $[0, \Delta_{\max}]$ with
\begin{eqnarray}\label{eqnRevise_Algo_2_4}\Delta_{\max} = \frac{ \bar E - \sum_{\!j\!=\!K_1+1}^{\!K} tr ( \tilde{\bf H}_{12}^{(j)} {\bf
Q}_j (\tilde{\bf H}_{12}^{(j)})^H) - {\boldsymbol \omega}^T {\bf P}^{[n]}}{{\boldsymbol \omega}^T \nabla_{{\bf P}} J^{UP}({\bf P}^{[n]}, {\bf Q}_{K_1+1}^{[n]},...,{\bf Q}_K^{[n]}) },
\end{eqnarray}
where ${\boldsymbol \omega} = [\omega_1,...,\omega_{K_1}]^T$. Then, update
$E_{11}^{[n+1]}$ and ${\bf R}_{-i}^{[n+1]}$ with ${\bf P}^{[n+1]}$
similarly to (\ref{eqnRevise_Algo_2_1}).
\end{enumerate}
\item Finally, the boundary point of the achievable R-E region is given as
 \begin{eqnarray}\label{eqnRevise_Algo_2_5}
R &=&\sum_{i=K_1+1}^{K}\log \det({\bf I}_{M} + {\bf H}_{ii}^H(\bar{\bf R}_{-i}^{[N_{max}+1]})^{-1}{\bf
H}_{ii}{\bf Q}_i^{[N_{max}+1]} ),\nonumber\\\nonumber
E& =& E_{11}^{[N_{max}+1]} + \sum_{\!j\!=\!K_1+1}^{\!K} tr ( \tilde{\bf H}_{12}^{(j)} {\bf
Q}_j^{[N_{max}+1]} (\tilde{\bf H}_{12}^{(j)})^H).
\end{eqnarray}
\end{enumerate}
 \vspace*{5pt}
In (\ref{eqnRevise_Algo_2_4}), if the total transferred energy is larger than the required harvested energy $\bar
E$, the transmitters transferring the energy reduce their transmit power to lower the interference to the ID receivers. Furthermore,
the maximum allowable step size in (\ref{eqnRevise_Algo_2_4}) is computed from (\ref{eqnRevise_Algo_2_3}) and the fact that $E_{11}^{[n+1]} = \bar E -\sum_{\!j\!=\!K_1+1}^{\!K} tr ( \tilde{\bf H}_{12}^{(j)} {\bf Q}_j^{[n]} (\tilde{\bf H}_{12}^{(j)})^H) $, which leads
 \begin{eqnarray}\label{REregion_14}
E_{11}^{[n+1]} = {\boldsymbol \omega}^T {\bf P}^{[n+1]} = {\boldsymbol \omega}^T {\bf P}^{[n]} + \Delta_{\max}\cdot {\boldsymbol \omega}^T\nabla_{{\bf P}} J^{UP}({\bf P}^{[n]}, {\bf Q}_{K_1+1}^{[n]},...,{\bf Q}_K^{[n]}).
\end{eqnarray}
Note that, if the
energy harvested by the EH receivers from the information transmitters
($\sum_{\!j\!=\!K_1+1}^{\!K} tr ( \tilde{\bf H}_{12}^{(j)} {\bf
Q}_j^{[n]} (\tilde{\bf H}_{12}^{(j)})^H)$) is larger
than $\bar E$, the energy transmitters do not transmit any signal.
That is, $rank({\bf Q}_k) =0$.

To complete Algorithm 1, we now show how to solve the optimization
problem (P1-UP) for ${\bf Q}_k^{[n]}$, $k=K_1+1,...,K$ in Step 2 of Algorithm 1.
For given $E_{11}$ and ${\bf Q}_i$, $i=1,...,K_1$, $ J^{UP}({\bf P}^{[n]}, {\bf Q}_{K_1+1}^{[n]},...,{\bf Q}_K^{[n]})$ can be derived as
 \begin{eqnarray}\label{REregion_15}\nonumber
J^{UP}&=&\sum_{i=K_1+1}^{K}\log \det({\bf I}_{M} + {\bf H}_{ii}^H(\bar{\bf R}_{-i}^{[n]})^{-1}{\bf
H}_{ii}{\bf Q}_i )\\&=&\sum_{i=K_1+1}^{K}\log \det({\bf I}_{M} +(\bar{\bf R}_{-i}^{[n]})^{-1/2}{\bf
H}_{ii}{\bf Q}_i {\bf H}_{ii}^H(\bar{\bf R}_{-i}^{[n]})^{-1/2}).
\end{eqnarray}
By letting $\bar{\bf H}_{ii} = (\bar{\bf R}_{-i}^{[n]})^{-1/2}{\bf H}_{ii}$, the Lagrangian function of (P1-UP) can then be written as
\begin{eqnarray}\label{REregion_16}\nonumber
\!&L({\bf Q}_{i}, \lambda, \mu_i, i=K_{1}+1,...,K) = \sum_{i=K_1+1}^{K}\log \det({\bf I}_{M} +\bar{\bf
H}_{ii}{\bf Q}_i \bar{\bf H}_{ii}^H)\quad\quad\quad\quad\quad\quad\quad&
\\&  +\lambda( \sum_{\!i\!=\!K_1+1}^{\!K} tr ( \tilde{\bf H}_{12}^{(i)} {\bf
Q}_i (\tilde{\bf H}_{12}^{(i)})^H) -
(\bar E \!-\! E_{11})  ) - \sum_{i=K_1+1}^{K} \mu_{i} (tr({\bf Q}_{i}) - P),&
\end{eqnarray}
and the corresponding dual function is then given by \cite{SBoyd,
Zhang1}
\begin{eqnarray}\label{REregion_17}
g(\lambda, \mu_i, i= K_1+1,...,K) = \underset{{\bf Q}_i \succeq{\bf 0}}{\max} L({\bf Q}_{i},\lambda, \mu_i, i= K_1+1,...,K ).
\end{eqnarray}
Here the optimal solution $\mu_i'$, $\lambda'$, and ${\bf Q}_i$ can
be found through the iteration of the following steps \cite{SBoyd}

\vspace*{5pt}Algorithm 2. {\it{\underline{Optimization algorithm for P1-UP given ${\bf P}$:}}}
\begin{enumerate}
\item The maximization of $L({\bf Q}_{i}, \lambda, \mu_i, i = K_1+1,...,K)$ over ${\bf Q}_i$ for given $\lambda, \mu_i$.
\item The minimization of $g(\lambda, \mu_i, i= K_1+1,...,K) $ over $\lambda, \mu_i$ for given ${\bf Q}_i$.
\end{enumerate}
\vspace*{5pt}
Note that, for given $\lambda, \mu_i$, $i=K_1,...,K$, the maximization of $L({\bf
Q}_2, \lambda, \mu)$ in Step 1) can be derived as
\begin{eqnarray}\label{REregion_18}
\!\underset{{\bf Q}_i \succeq{\bf 0}}{\max} L({\bf Q}_{i}, \lambda, \mu_i, i = K_1+1,...,K) &=& \sum_{i=K_1+1}^{K}\log \det({\bf I}_{M} +\bar{\bf
H}_{ii}{\bf Q}_i \bar{\bf H}_{ii}^H) +\lambda( \sum_{\!i\!=\!K_1+1}^{\!K} tr ( \tilde{\bf H}_{12}^{(i)} {\bf
Q}_i (\tilde{\bf H}_{12}^{(i)})^H))\nonumber\\& & - \sum_{i=K_1+1}^{K} \mu_{i} (tr({\bf Q}_{i})),\nonumber\\
&=& \sum_{i=K_1+1}^{K}\left(\log \det({\bf I}_{M} +\bar{\bf
H}_{ii}{\bf Q}_i \bar{\bf H}_{ii}^H) -tr({\bf A}_i{\bf Q}_i)   \right),
\end{eqnarray}
where ${\bf A}_i= {\mu}_i{\bf I}_M -\lambda(\tilde{\bf H}_{12}^{(i)})^H\tilde{\bf H}_{12}^{(i)}$. Note that, due to the assumption that the interferences from other information transmitters is nulled out in (P1-UP), (\ref{REregion_18}) can be easily decoupled into the point-to-point
MIMO capacity optimization with a single weighted power constraint
and the solution for ${\bf Q}_i$ is then given by \cite{SBoyd, Zhang1}
\begin{eqnarray}\label{REregion_19}
{\bf Q}_i &=& {\bf A}_i^{-1/2}\bar{\bf V}'_{ii}\bar{\boldsymbol
\Lambda}_i'\bar{\bf V}_{ii}'^H{\bf A}_i^{-1/2},
\end{eqnarray}
where $\bar{\bf V}'_{ii}$ is obtained from the SVD of the matrix
$\bar{\bf H}_{ii}{\bf A}_i^{-1/2}$, i.e., $\bar{\bf H}_{ii}{\bf
A}_i^{-1/2} =\bar{\bf U}'_{ii}\bar{\boldsymbol \Sigma}'_{ii}
\bar{\bf V}_{ii}'^H$. Here, $\bar{\boldsymbol \Sigma}'_{ii} =
diag\{ \sigma_{1}(\bar{\bf H}_{ii}{\bf
A}_i^{-1/2}),...,\sigma_{M}(\bar{\bf H}_{ii}{\bf
A}_i^{-1/2}) \}$ and
$\bar{\boldsymbol \Lambda}_i' = diag\{ \bar p_{i,1},...,\bar p_{i,M}
\}$ with $\bar p_{i,j} = (1-1/\sigma_{j}^2(\bar{\bf H}_{ii}{\bf
A}_i^{-1/2}))^+ $,
$j=1,...,M$.

In step 2), the parameters $\mu_i$ and $\lambda$ minimizing
$g(\lambda, \mu_i, i= K_1+1,...,K)$ can be solved by the subgradient-based
method \cite{Zhang1, XZhao}, where the the subgradient of
$g(\lambda, \mu_i, i= K_1+1,...,K)$ is given by
\begin{eqnarray}\label{REregion_20}
\nabla g(\lambda, \mu_i, i= K_1+1,...,K) =\biggl( \sum_{\!i\!=\!K_1+1}^{\!K} tr ( \tilde{\bf H}_{12}^{(i)} {\bf
Q}_i (\tilde{\bf H}_{12}^{(i)})^H) -
(\bar E \!-\! E_{11}) ,  \ P - tr({\bf Q}_{K_1+1}), ...,P - tr({\bf Q}_{K})\biggr).
\end{eqnarray}

\begin{remark}\label{remark2nd_revised}
Due to the fact that each element in (\ref{REregion_13}) always has a negative value and the step size in (\ref{eqnRevise_Algo_2_4}) has non-negative values, we can find that the power of the energy transmitters converges monotonically. In addition, because (\ref{REregion_10}) is concave over ${\bf Q}_i$ and monotonically decreasing with respect to ${\bf P}$, we can easily find that every superlevel set $\{{\bf Q}_i, {\bf P}| J^{UP}({\bf Q}_i,i=K_1+1,...,K, {\bf P}) \geq \alpha\}$ for $\alpha \in \mathbb{R}$ is convex. That is, from the definition of the quasi-concavity (Section 3.4.1 of \cite{SBoyd}), (\ref{REregion_10}) is quasi-concave. Furthermore, the
constraints, (\ref{REregion_11}), (\ref{REregion_12}) and
(\ref{REregion_12_1}), are the convex set of ${\bf Q}_i$ and
$P_j$. Therefore, the converged solution of Algorithm 1 is
globally optimal \cite{Bazaraa_book}. Note that, given ${\bf P}$ (that monotonically converges), $J^{UP}$ in (\ref{REregion_15}) is concave and satisfies the slater's condition \cite{SBoyd}, it has a zero duality gap.
\end{remark}

Now let us consider the original problem (P1). Because (P1) is non-convex, the optimal solution cannot be easily computed, but motivated by Algorithm 1, we can also develop a sub-optimal iterative algorithm. Note that $ \nabla_{{\bf P}} J({\bf P}, {\bf Q}_{K_1+1},...,{\bf Q}_K) \in \mathbb{R}^{K_1 \times 1}$ has the same form as $ \nabla_{{\bf P}} J^{UP}({\bf P}, {\bf Q}_{K_1+1},...,{\bf Q}_K)$ except that $\bar{\bf R}_{-i}$ is replaced by ${\bf R}_{-i}$. Therefore, the objective function in (P1) is also monotonically decreasing with respect to ${\bf P}$, regardless of ${\bf R}_{-i}$, and ${\bf P}$ for (P1) can also be optimized using the steepest descent method as Algorithm 1, in which $\bar{\bf R}_{-i}^{[n]}$ is replaced by
\begin{eqnarray}\label{REregion_21}
{\bf R}_{-i}^{[n]} = {\bf I}_{M}+ \sum_{ j=1}^{K_1}({\bf P}^{[n]})_j{\bf \Omega}_{ij} + \sum_{\substack{ j=K_1+1 \\
j\neq i}}^{K}{\bf H}_{ij}{\bf Q}_j^{[n]}{\bf H}_{ij}^H.
\end{eqnarray}
In addition, we formulate the Lagrangian function similarly to (\ref{REregion_16})
\begin{eqnarray}\label{REregion_21_1}\nonumber
\!&L'({\bf Q}_{i}, \lambda, \mu_i, i=K_{1}+1,...,K) = \sum_{i=K_1+1}^{K}\log \det({\bf I}_{M} + {\bf H}_{ii}^H({\bf R}_{-i}^{[n]})^{-1}{\bf
H}_{ii}{\bf Q}_i )\quad\quad\quad\quad\quad\quad\quad&
\\\nonumber&  +\lambda( \sum_{\!i\!=\!K_1+1}^{\!K} tr ( \tilde{\bf H}_{12}^{(i)} {\bf
Q}_i (\tilde{\bf H}_{12}^{(i)})^H) -
(\bar E \!-\! E_{11})  ) - \sum_{i=K_1+1}^{K} \mu_{i} (tr({\bf Q}_{i}) - P),&
\end{eqnarray}
and the corresponding dual function is then given by
\begin{eqnarray}\label{REregion_17_1}
g'(\lambda, \mu_i, i= K_1+1,...,K) = \underset{{\bf Q}_i \succeq{\bf 0}}{\max} L'({\bf Q}_{i},\lambda, \mu_i, i= K_1+1,...,K ).
\end{eqnarray}
Note that the only difference between $L'({\bf Q}_{i}, \lambda, \mu_i, i=K_{1}+1,...,K) $ and $L({\bf Q}_{i}, \lambda, \mu_i, i=K_{1}+1,...,K) $ is the interference at the $k$th ID receiver due to the other information transmitters, which hinders finding the optimal solution for the maximization of $L'({\bf Q}_{i}, \lambda, \mu_i, i=K_{1}+1,...,K) $ over ${\bf Q}_i$ for given $\lambda$ and $\mu_i$. That is, for given $\lambda, \mu_i$, $i=K_1,...,K$, the maximization of $L'({\bf
Q}_2, \lambda, \mu)$, analogous with Step 1) of Algorithm 2, can be derived as
\begin{eqnarray}\label{REregion_18_1}
\!\!\underset{{\bf Q}_i \succeq{\bf 0}}{\max} L'({\bf Q}_{i}, \lambda, \mu_i, i = K_1+1,...,K) &\!\!=\!\!& \sum_{i=K_1+1}^{K}\log \det({\bf I}_{M}+ {\bf H}_{ii}^H({\bf R}_{-i}^{[n]})^{-1}{\bf
H}_{ii}{\bf Q}_i ) \nonumber\\\!\!&\!\!\!\! & +\lambda( \sum_{\!i\!=\!K_1+1}^{\!K} tr ( \tilde{\bf H}_{12}^{(i)} {\bf
Q}_i (\tilde{\bf H}_{12}^{(i)})^H)) - \sum_{i=K_1+1}^{K} \mu_{i} (tr({\bf Q}_{i})),\nonumber\\
\!\!&\!\!=\!\!& \!\sum_{i=K_1+1}^{K}\left(\log \det({\bf I}_{M}+ {\bf H}_{ii}^H({\bf R}_{-i}^{[n]})^{-1}{\bf
H}_{ii}{\bf Q}_i ) -tr({\bf A}_i{\bf Q}_i)   \right)\!.\!
\end{eqnarray}
Letting $ \hat{\bf Q}_i = {\bf A}_i^{1/2}{\bf Q}_i{\bf A}_i^{1/2}$ similarly to \cite{ZhangLiangCui1},
(\ref{REregion_18_1}) can be rewritten as
\begin{eqnarray}\label{REregion_19_1}
\!\underset{{\bf Q}_i \!\succeq\!{\bf 0}}{\max} L'({\bf Q}_{i}, \lambda, \mu_i, i\! =\! K_1\!+\!1,...,K) \!=\!\underset{\hat{\bf Q}_i\! \succeq\!{\bf 0}}{\max} \!\sum_{i=K_1\!+\!1}^{K}\!\left(\!\log\! \det({\bf I}_{M} + {\bf A}_i^{-\!1/2}{\bf H}_{ii}^H({\bf R}_{-i}^{[n]})^{-\!1}{\bf
H}_{ii}{\bf A}_i^{-\!1/2}\hat{\bf Q}_i ) -tr(\hat{\bf Q}_i) \!  \right),
\end{eqnarray}
which becomes a conventional rate maximization problem in MIMO IFC subject to individual power constraints \cite{NegroShenoy, Scutari}. Accordingly, for (P1), ${\bf Q}_i$, $\lambda$, and $\mu_i$ can be optimized based on the iterative waterfilling algorithm \cite{Scutari} with effective channel ${\bf
H}_{ii}{\bf A}_i^{-1/2}$. %The iterative algorithm for P1 is summarized in Algorithm 3.
Note that, even though the iterative waterfilling cannot achieve the global optimum for the non-convex (\ref{REregion_19_1}), its convergence to the Nash equilibrium (local optimum) is guaranteed for the nonsingular channel matrices ${\bf H}_{kk},k=1,...,K$ \cite{Scutari, YeBlum}. We also note that, recently, the convergence to the global optimum can be achieved by the global optimization methods such as the {\it{difference of two convex functions (D.C.) programming}} \cite{AlShatri, XuLeNgoc}. However, they would require a {\it{centralized}} optimization process (i.e., the explicit coordination among the nodes and the complete knowledge of all channel responses which are not available in our system model).
Furthermore, the objective function in (P1) is also monotonically decreasing with respect to ${\bf P}$, regardless of ${\bf R}_{-i}$, and accordingly, we can also find that the power of energy transmitters converges monotonically. Therefore, the convergence of the proposed iterative algorithm for (P1) is guaranteed.

\begin{remark}\label{remark_largenumberofIDs}
In (P1), if the number of information transmitters becomes large, the interference from energy transmitters is smaller than that from information transmitters. That is, (\ref{REregion_4}) can be approximated as
\begin{eqnarray}\label{Rem_largenumberID}
{\bf R}_{-i} \approx {\bf I}_{M}+  \sum_{\substack{ j=K_1+1 \\
j\neq i}}^{K}{\bf H}_{ij}{\bf Q}_j{\bf H}_{ij}^H,
\end{eqnarray}
and the achievable rate in (\ref{REregion_1}) is independent of the interference from energy transmitters. That is, the energy transmitter signals can be designed by caring about their own links, not caring about
the interference link to the ID receivers. Accordingly, as the number of information transmitters increases, the optimal transmission strategy at energy transmitters becomes a rank-one MEB method with a power control.
\end{remark}

%\subsection{Asymptotic behavior for a large $M$}\label{ssec:Asymptotic}
%Because of power efficiency and a low decoding complexity
%achieving near-optimal performance in information transfer, a
%massive MIMO system becomes a promising future communication
%structure. Accordingly, to get an insight on the joint information
%and energy transfer with a large number of antennas, we revisit
%the transmission strategy in Section \ref{ssec:optimal_tx_strat}
%for a large $M$.
%
%
%\begin{lem}\label{lem2} When $M$ goes to infinity, ${\bf R}_{-2}$
%approaches
% $\blacksquare$
%\end{lem}

\subsection{$K_1$-EH Selection algorithm in K-user MIMO IFC}\label{ssec:userselection}
Motivating that the SLER value indicates how suitable the current channel is to either EH mode or ID mode, we propose $K_1$-EH selection method in K-user MIMO IFC. That is, higher SLER implies that the transmitter can transfer more energy to its associated EH receiver and/or incur less interference to the ID receiver. Note that the $k$th SLER in (\ref{GSVD1}) depends on $(K_1 -1 )$ EH receivers. Therefore, in our proposed selection, to choose $K_1$ transceiver pairs jointly, we evaluate sum of SLERs of $K_1$ transceiver pairs for $_{K}C_{K_1}$ possible candidates and choose one candidate having the maximum SLER sum. That is, by letting
\begin{eqnarray}\label{userselect1}
I_j = \{ (Tx_{j_1}, Rx_{j_1}),...,(Tx_{j_{K_1}}, Rx_{j_{K_1}})| j_{k} \in \{1,...,K\}, {\text{ for }} k=1,...,K_1\},
\end{eqnarray}
the set of $K_1$ energy harvesting EHs is then selected as:
\begin{eqnarray}\label{userselect2}
I_{\max} &=&\underset{I_j}{\arg}\max \sum_{i\in I_j} SLER_i.
\end{eqnarray}

\section{Simulation Results}
\label{sec:simulation}

Computer simulations have been performed to evaluate the R-E
tradeoff of various transmission strategies in the K-user MIMO
IFC. In the simulations, the normalized channel ${\bf H}_{ij}$ is
generated such as ${\bf H}_{ij}=10^{-3/2}
\frac{\sqrt{\alpha_{ij}M}}{\|\tilde{\bf H}_{ij}\|_F}\tilde{\bf
H}_{ij}$, where the elements of $\tilde{\bf H}_{ij}$ are
independent and identically distributed (i.i.d.) zero-mean complex
Gaussian random variables (RVs) with a unit variance. The term
$10^{-3/2}$ is due to the path loss with a power path loss
exponent $3$ and $10m$ distance between Tx $i$ and Rx $i$ ($-30dB
= 10\log_{10}10^{-3}$). The maximum transmit power is set as
$P=50mW$ and the noise power is $1 \mu W$, unless otherwise
stated.

 Figs. \ref{Fig:R_E_tradeoff_multipleEH} (a)-(c) show R-E tradeoff curves for three different rank-one beamforming -
the MEB, the MLB, and the SLER maximizing beamforming described in Section \ref{ssec:rankoneBF} with $M =
4$ when $K=\{2,3,4\}$ and $K-K_1 =1$. Here, the path loss scale factor is set as $\alpha_{ii} = 1$, and $\alpha_{ij} = 0.6$ for $i,j=1,2,...,K$ and
$i \neq j$. That is, while only the $K$th transmitter transfers information to its corresponding receivers, the remaining transmitters transfer the energy to the remaining EH receivers with rank-one beamforming. As expected, as the number of energy transmitters increases, total harvested energy increases. Interestingly, in the regions where the energy is less
than a certain threshold $[40,100]\mu W$, all the energy transmitters
do not transmit any signals to reduce the interference to the
ID receiver. That is, the energy transferred from the
information transmitter is sufficient to satisfy the energy constraint
at the EH receivers. Note that the threshold is linearly proportional to the number of EH receivers that can harvest energy from information transmitter's signal. The dashed lines in  Figs. \ref{Fig:R_E_tradeoff_multipleEH} (a)-(c) indicate the R-E curves of the time-sharing of
the full-power rank-one beamforming and the no transmission at the energy transmitters. Accordingly, the information
transmitter switches between the beamforming on $\tilde{\bf H}_{12}^{(K)}$ and the water-filling on ${\bf H}_{KK}$
in the corresponding time slots. For MLB, ``water-filling-like'' approach (\ref{REregion_19})
exhibits higher R-E performance than the time-sharing scheme.
However, for MEB, when the required energy is less than a certain value, the
time-sharing exhibits better performance than the approach
(\ref{REregion_19}). This observation is more apparent as the number of energy transmitters increases.
That is, because the energy transmitters with the MEB cause large
interference to the ID receiver, it is desirable that, for the low
required harvested energy, the energy transmitters turn off their
power in the time slots where the information transmitter is assigned
to exploit the water-filling method on ${\bf H}_{KK}$. Instead, in the remaining time slots,
the energy transmitters opt for a MEB with full power and the
information transmitter transfers its information to the ID receiver by
steering its beam on EH receiver's channel $\tilde{\bf H}_{12}^{(K)}$ to help the EH operation.
In contrast, with SLER maximizing beamforming, ``water-filling-like'' approach (\ref{REregion_19})
exhibits higher R-E performance than the time-sharing scheme. Furthermore, its R-E region covers most of those of both MEB and
MLB.

\begin{figure}%[htbp]
\centering %\hspace{-3em}
 \subfigure[]
  {\includegraphics[height=4.6cm]{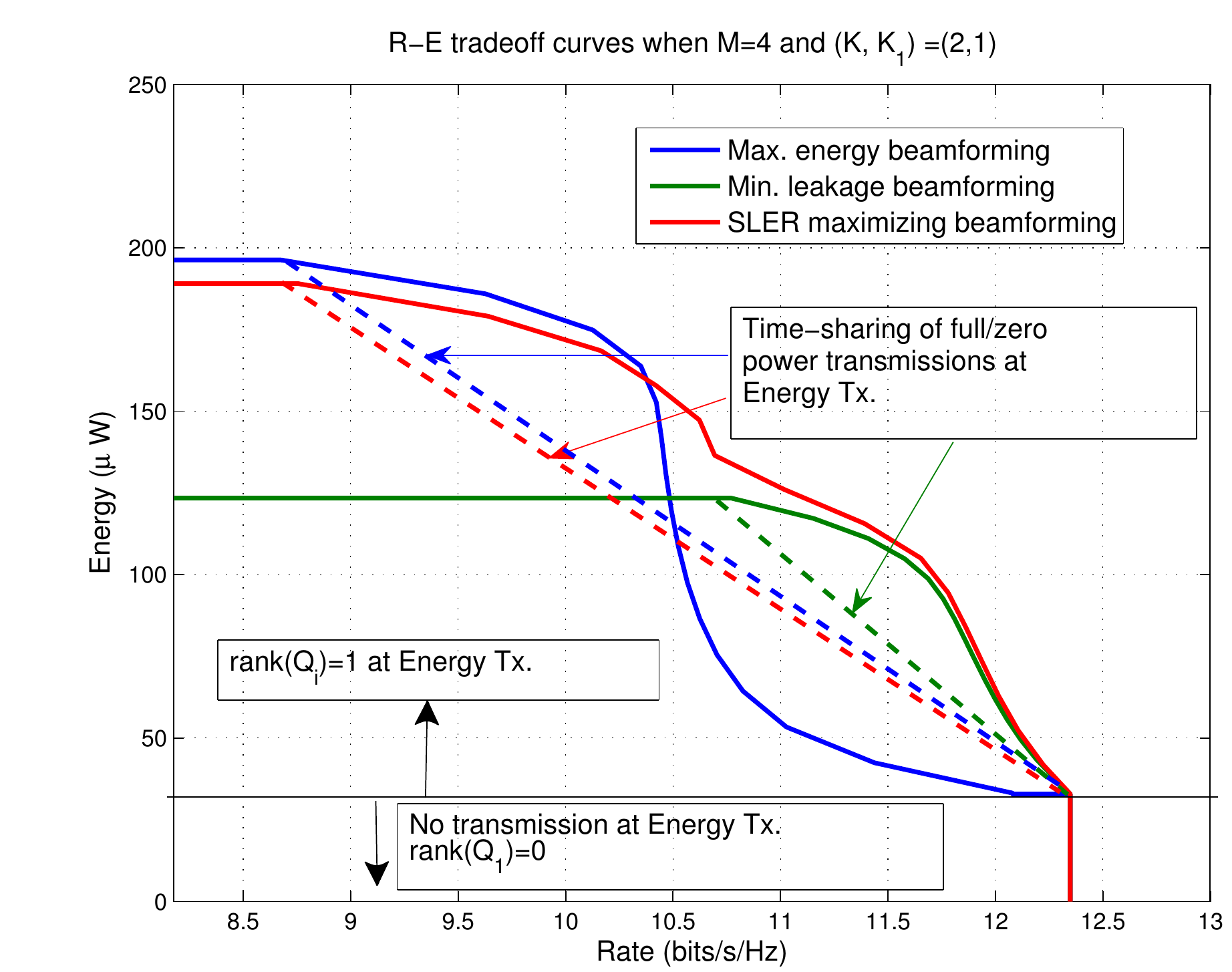}}
 \subfigure[]
  {\includegraphics[height=4.6cm]{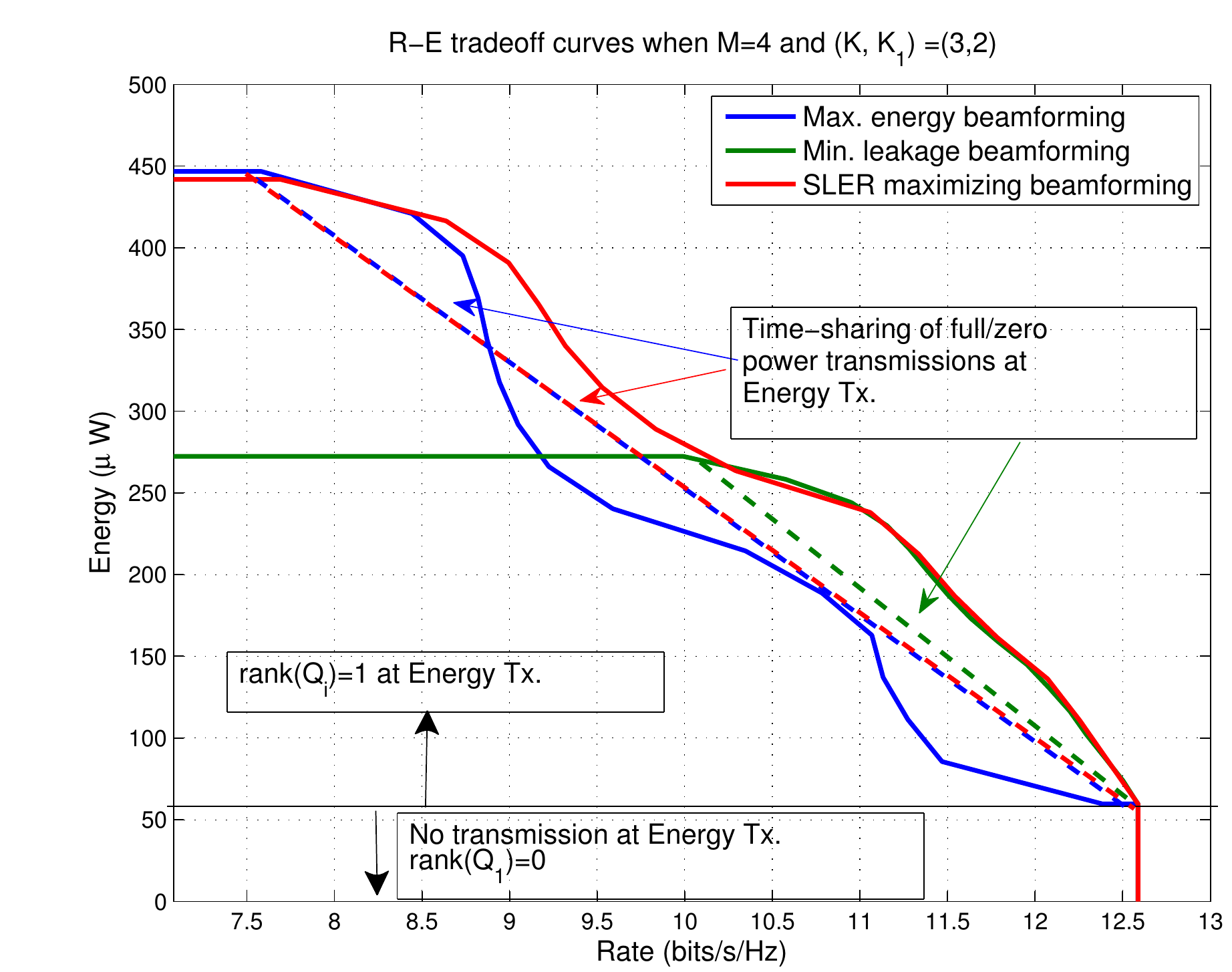}}
   \subfigure[]
  {\includegraphics[height=4.6cm]{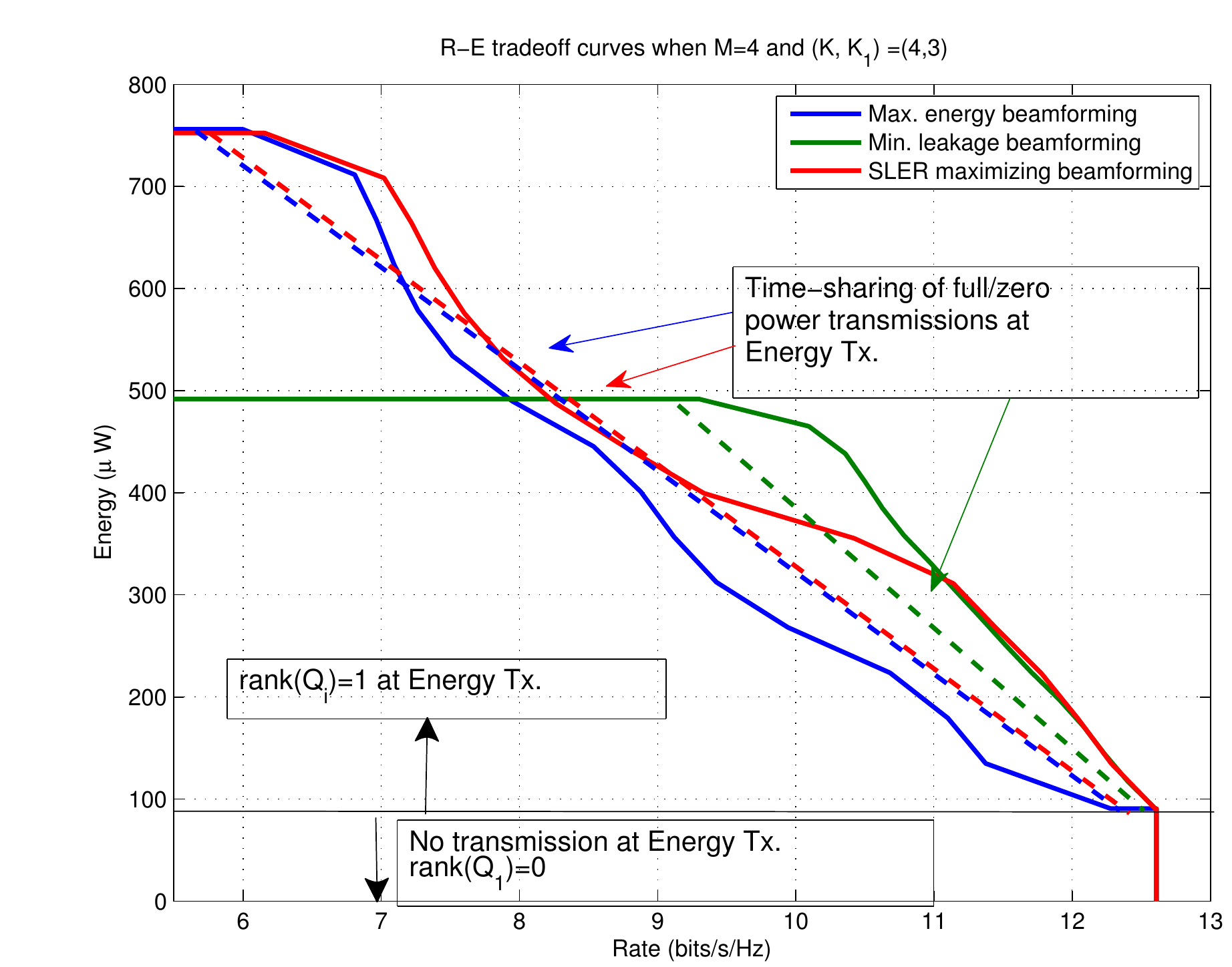}}
 \caption{R-E tradeoff curves for
MEB, MLB, and SLER maximizing beamforming when (a) $(K,K_1) =(2,1)$, (b) $(K,K_1) =(3,2)$, and (c) $(K,K_1) =(4,3)$. Here, $M=4$ and $\alpha_{ij}= 0.6$ for $i\neq
j$.} \label{Fig:R_E_tradeoff_multipleEH}
\end{figure}

In Figs. \ref{Fig:R_E_tradeoff_multipleID} (a) and (b), we have additionally
included the R-E tradeoff curves when $(K, K_1)= (3, 1)$ and $(4,1)$. Together with Fig. \ref{Fig:R_E_tradeoff_multipleEH} (a), we can find the trend of the R-E region when
the number of information transmitters increases, while only the first transmitter transfers energy. Note that, as the number of information transmitters increases, the maximum harvested energy is also increased. However, because the interferences are increasing proportional to the number of information transmitters, the maximum information rate is not drastically increasing, which implies that the system becomes interference-limited. Furthermore, because the interference due to the information transmitters is dominant, MEB at the energy transmitter becomes a more attractive strategy, resulting in a wider R-E region compared to that for the MLB. That is, compared to Fig. \ref{Fig:R_E_tradeoff_multipleEH}, the R-E region of MEB covers almost that of MLB (see Fig. \ref{Fig:R_E_tradeoff_multipleID} (b)), which is consistent with Remark \ref{remark_largenumberofIDs}.

To see the effect of interference on the information rate, we evaluate the R-E tradeoff curves in Fig. \ref{Fig:R_E_tradeoff_multipleEHmultipleID} when $(K, K_1)= (4, 2)$ (multiple EHs and multiple IDs) with different $\alpha_{ij} = \{0.6, 0.3\}$ for $i\neq j$. Because a smaller $\alpha_{ij}$ implies less interference at each receiver, we can find that R-E region for $\alpha_{ij} = 0.6$ exhibits a larger harvested energy but a lower information rate.
\begin{figure}%[htbp]
\centering %\hspace{-3em}
 \subfigure[]
  {\includegraphics[height=4.6cm]{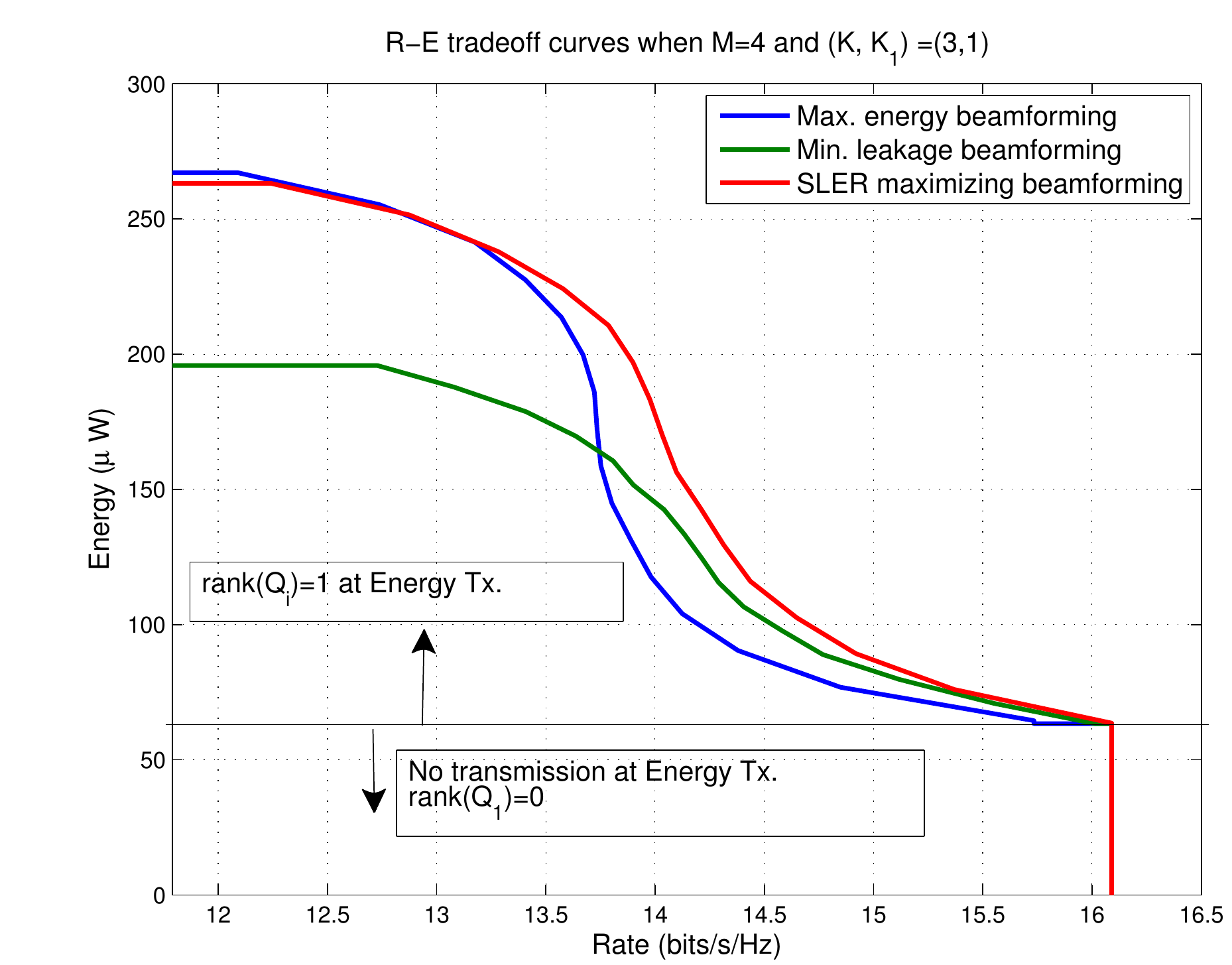}}
 \subfigure[]
  {\includegraphics[height=4.6cm]{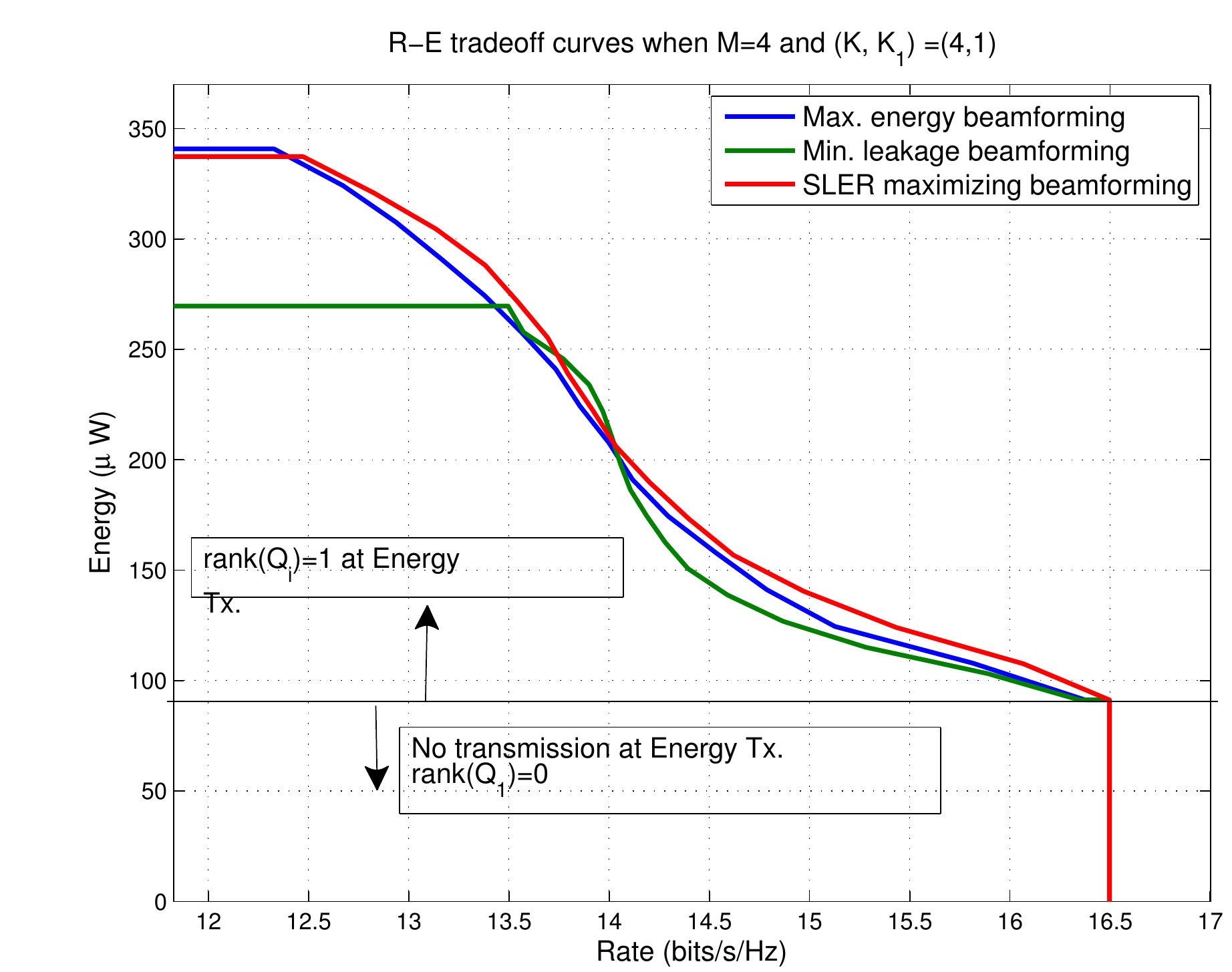}}
 \caption{R-E tradeoff curves for
MEB, MLB, and SLER maximizing beamforming when (a) $(K,K_1) =(3,1)$ and (b) $(K,K_1) =(4,1)$. Here, $M=4$ and $\alpha_{ij}= 0.6$ for $i\neq
j$.} \label{Fig:R_E_tradeoff_multipleID}
\end{figure}

\begin{figure}%[htbp]
\centering %\hspace{-3em}
 \subfigure[]
  {\includegraphics[height=4.6cm]{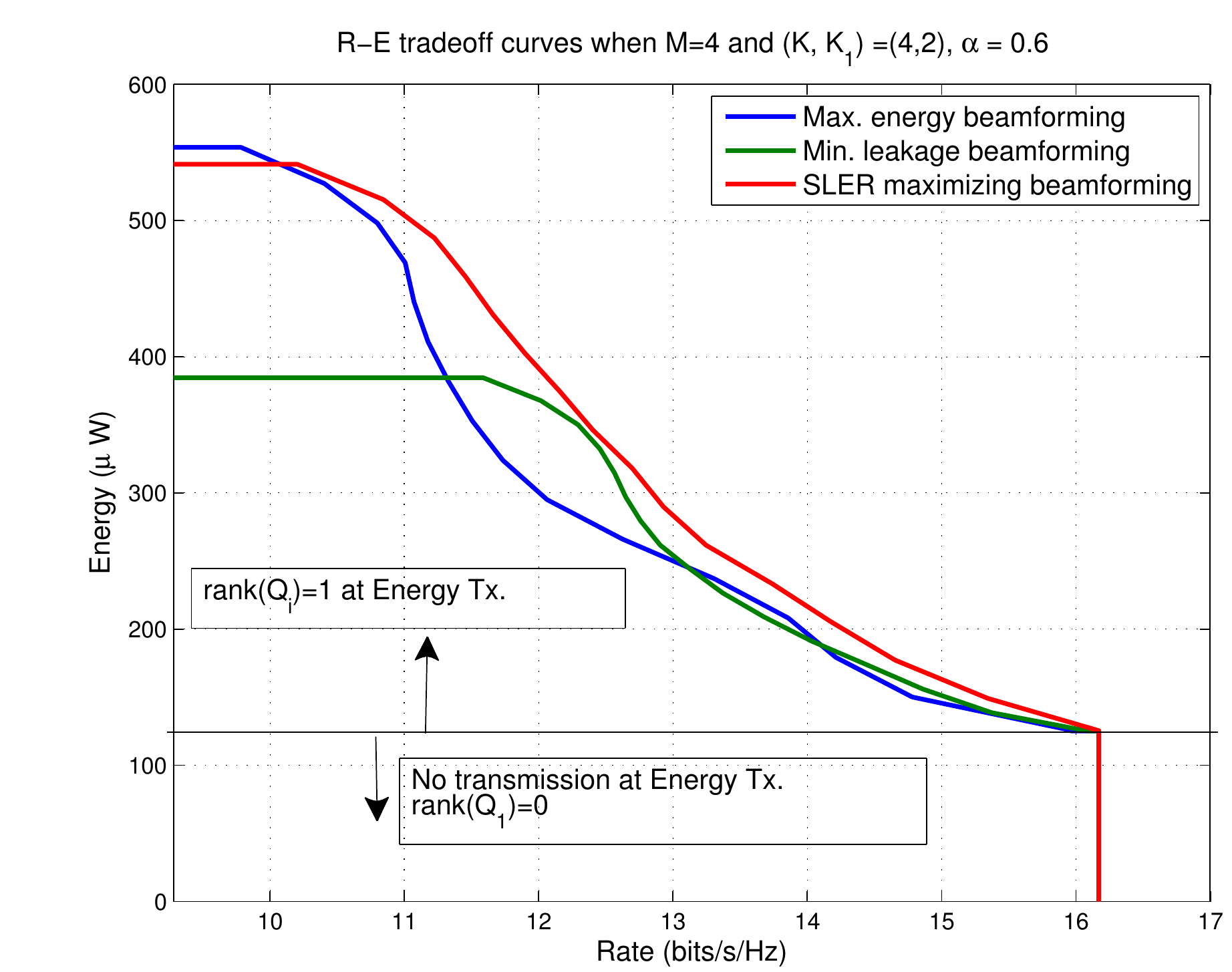}}
 \subfigure[]
  {\includegraphics[height=4.6cm]{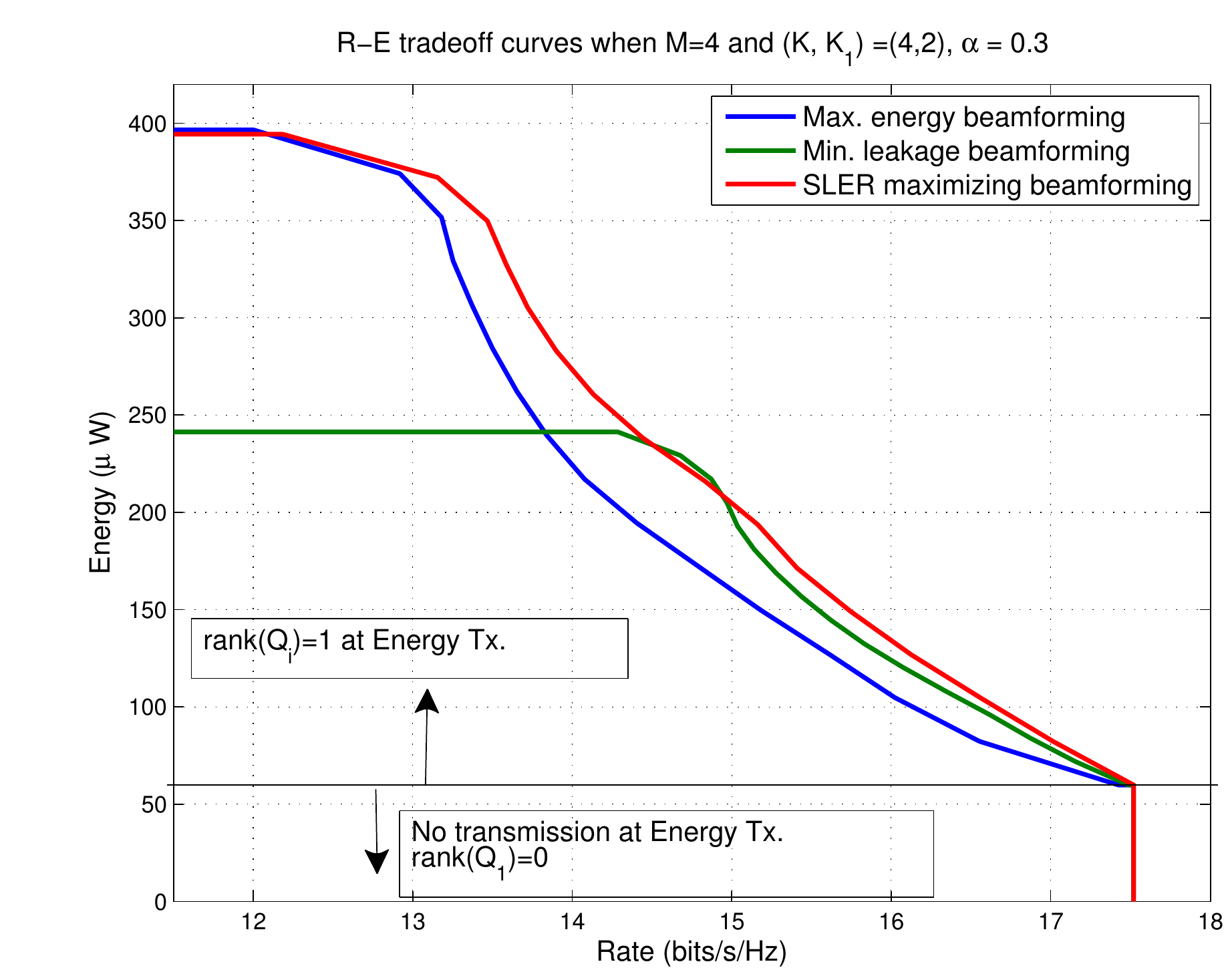}}
 \caption{R-E tradeoff curves for
MEB, MLB, and SLER maximizing beamforming when (a) $(K,K_1) =(4,2)$ and $\alpha_{ij}= 0.6$ for $i\neq
j$ and (b) $(K,K_1) =(4,2)$ and $\alpha_{ij}= 0.3$ for $i\neq
j$. Here, $M=4$.} \label{Fig:R_E_tradeoff_multipleEHmultipleID}
\end{figure}

Fig. \ref{R_E_tradeoff_GSVD_scheduling} shows the R-E tradeoff
curves for SLER maximizing beamforming with/without SLER-based
user selection described in Section \ref{ssec:userselection} when $(K,K_E)=(4,2)$ with $M=4$. Note that the
case with $\alpha_{ij} = 0.3 $ has weaker cross-link channel
(inducing less interference) than that with $\alpha_{ij} = 0.6 $.
The SLER-based user selection extends the achievable R-E region for
both $\alpha_{ij} \in \{0.3,0.6\}$, but the improvement for
$\alpha_{ij}=0.6$ is slightly more apparent. That is, the SLER-based
scheduling becomes more effective when strong interference exists
in the system. Note that the case with $\alpha_{ij} = 0.6 $ exhibits
a slightly lower achievable rate than that with $\alpha_{ij}=0.3$,
while achieving a larger harvested energy, which is a similar observation as that found in Fig. \ref{Fig:R_E_tradeoff_multipleEHmultipleID}. That is, a strong
interference degrades the information decoding performance but it
can be effectively utilized in the energy-harvesting.

\begin{figure}%[htbp]
\centering %\hspace{-3em}
 \subfigure[]
  {\includegraphics[height=4.6cm]{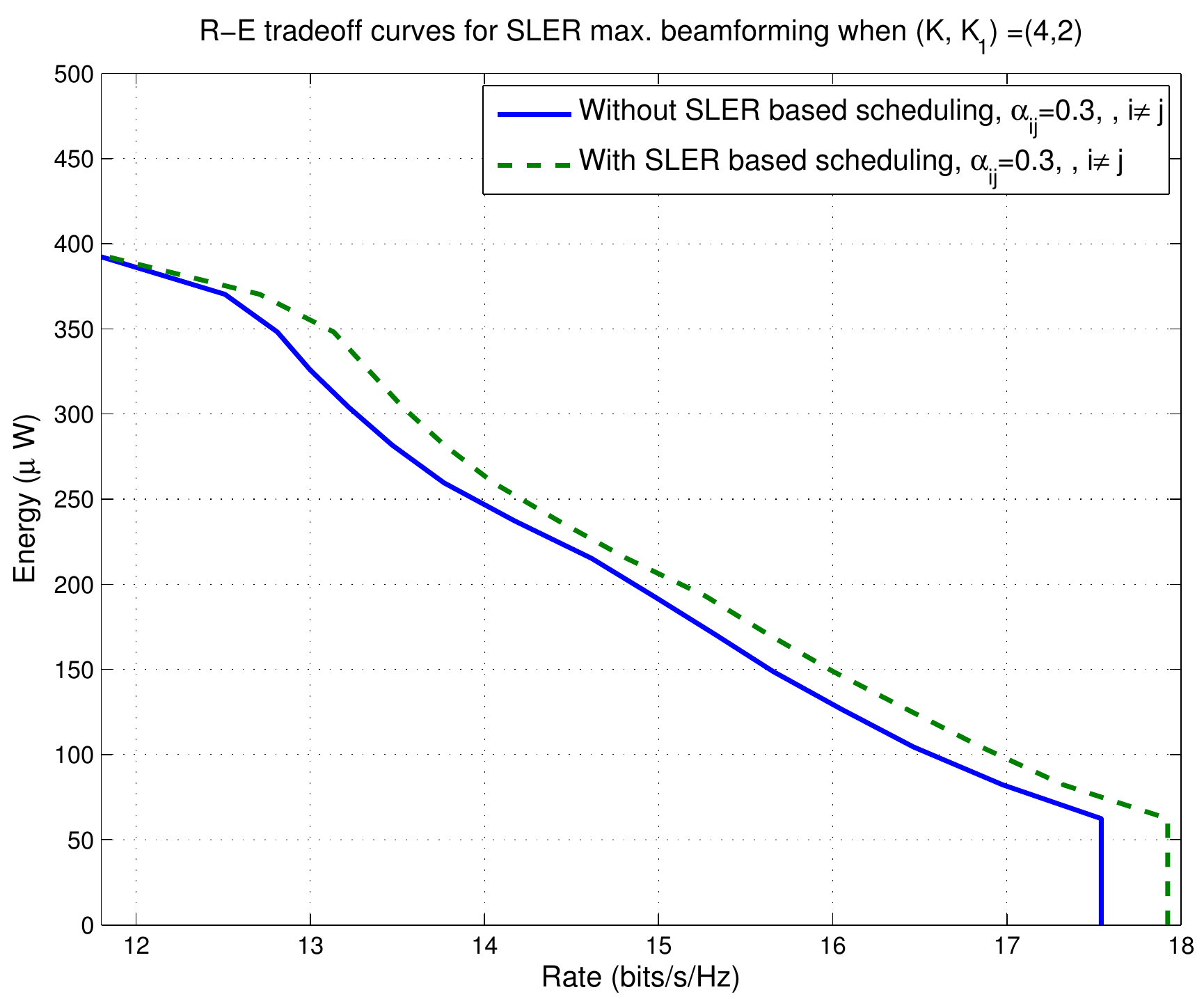}}
 \subfigure[]
  {\includegraphics[height=4.6cm]{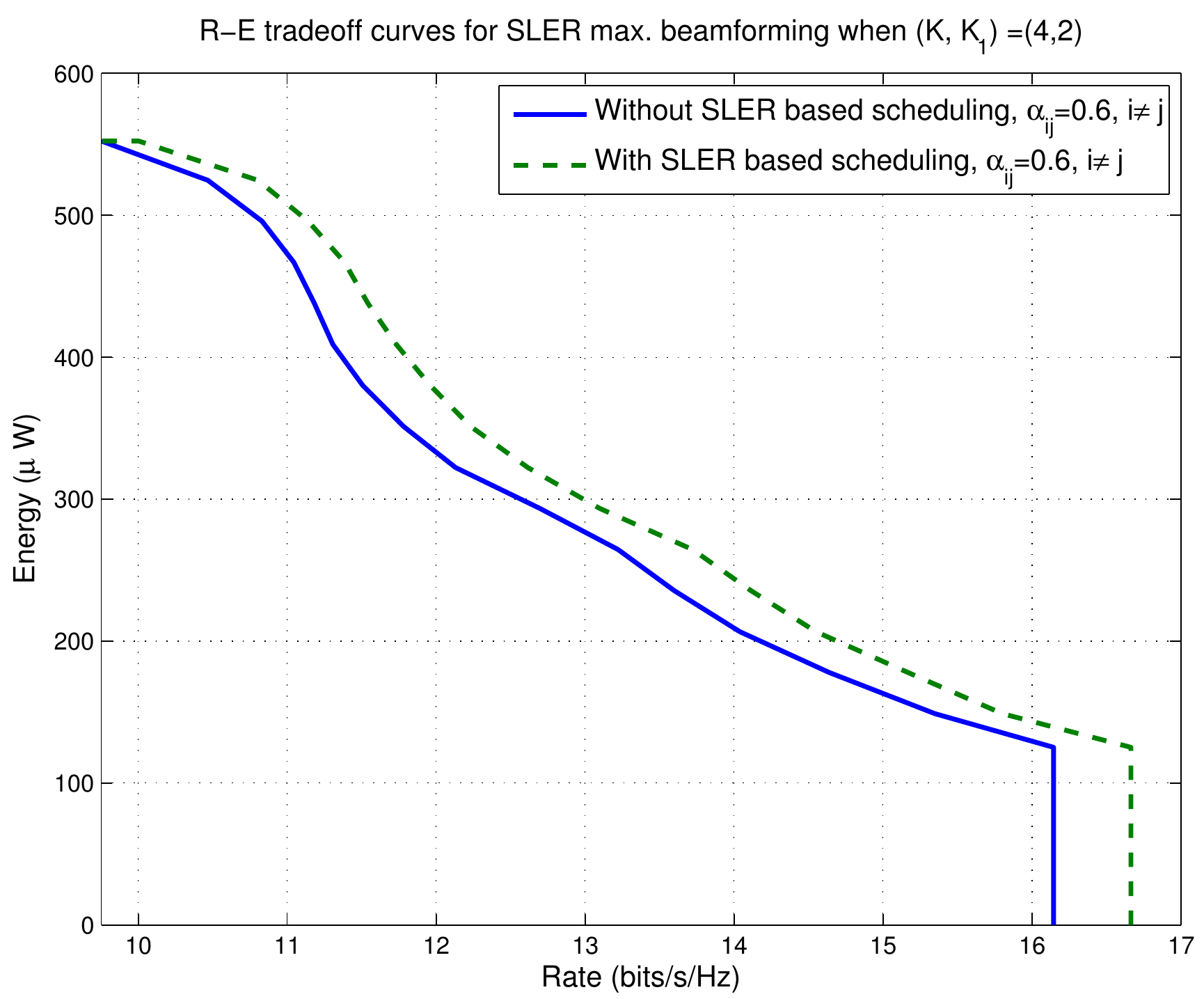}}
 \caption{   R-E tradeoff curves for
SLER maximizing beamforming with/without SLER-based user selection
when (a) $\alpha_{ij}= 0.3$ and (b) $\alpha_{ij}= 0.6$ for $i\neq
j$. Here, $M=4$.} \label{R_E_tradeoff_GSVD_scheduling}
\end{figure}

\begin{remark}\label{remark_largenumberofEHs}
In Fig. \ref{Fig:R_E_tradeoff_multipleEH} (c), SLER maximizing beamforming is outperformed by MLB at around a harvested energy of 450$\mu W$. This comes from the fact that the beam directions of energy transmitters are determined independently, because they do not share the CSIT.
Accordingly, even though the beam directions are determined to maximize their own SLERs, the aggregate interference may not be optimized (see also Remark \ref{remark_multiIDmultiEH}). Because the aggregate interference at the ID receivers is determined by both the directions and the powers of the energy beams, the information rate performance becomes increasingly sensitive to the beam steering and the power reduction as the number of energy transmitters increases. Therefore, in Fig. \ref{Fig:R_E_tradeoff_multipleEH} (c), the SLER maximizing beamforming with a small transmit power may be outperformed by the MLB with a large transmit power under the same harvesting energy (e.g., 450$\mu W$). Note that, in our simulations, we first fixed the beam directions of all energy transmitters as in Section \ref{ssec:rankoneBF}, and then reduce the powers of the energy transmitters in Algorithm 1. However, as stated in Remark \ref{remark_SLERBF1}, in the SLER maximizing beamforming, the beam direction softly bridges MEB and MLB depending on the scalar value multiplied by the identity matrix of the denominator in (\ref{GSVD1}). Therefore, if the harvested energy is enough in Step 2.b of Algorithm 1, before reducing the power of energy transmitters, we can tilt the beam to reduce the interference to the ID receivers by updating the energy beamforming vectors. Here, they can be updated by computing the GSVD of the matrix pair
\begin{eqnarray}\label{GSVD1_tilt}\nonumber
((\tilde{\bf
H}_{11}^{(k)})^H\tilde{\bf H}_{11}^{(k)}, (\tilde{\bf H}_{21}^{(k)})^H\tilde{\bf H}_{21}^{(k)} +\alpha^{n}{max(\bar E
/{K_1 P} -\|\tilde{\bf H}_{11}^{(k)}\|^2 ,0)}{\bf I}_M),
\end{eqnarray}
with a decaying factor $\alpha \in (0,1)$. In Fig. \ref{R_E_tradeoff_GSVD_tilt}, the R-E curve of a new SLER beamforming with beam tilting and power allocation is compared with that of SLER beamforming with only a power allocation when $(K, K_1)= (4,3), (5, 4)$. In our simulation, $\alpha$ is fixed as $0.9$.
%\begin{enumerate}
%\setcounter{enumii}{1}
%  \item  If $\sum_{\!j\!=\!K_1+1}^{\!K} tr ( \tilde{\bf H}_{12}^{(j)} {\bf
%Q}_j^{[n]} (\tilde{\bf H}_{12}^{(j)})^H) +
%E_{11}^{[n]}
%>\bar E$,
%\begin{enumerate}
%\item Letting $\bar E' = \bar E - \sum_{\!j\!=\!K_1+1}^{\!K} tr ( \tilde{\bf H}_{12}^{(j)} {\bf
%Q}_j^{[n]} (\tilde{\bf H}_{12}^{(j)})^H)$, if $  \bar E' > E_{11,0}$
%
%\item Else
% \begin{eqnarray}\label{eqnRevise_Algo_2_3} {\bf P}^{[n+1]} =max\left({\bf P}^{[n]} +\Delta\cdot \nabla_{{\bf P}} J^{UP}({\bf P}^{[n]}, {\bf Q}_{K_1+1}^{[n]},...,{\bf Q}_K^{[n]}), {\bf 0}\right).
%\end{eqnarray}
%\end{enumerate}
%\end{enumerate}
We can see that the new SLER beamforming scheme exhibits better performance than the SLER without beam tilting and the effect of beam tilting is more apparent for $(K, K_1) =(5, 4)$. These evaluations show that further beamforming enhancements are possible by better jointly designing beam directions and power in K-user MIMO IFC.
\begin{figure}%[htbp]
\centering %\hspace{-3em}
   \subfigure[]
  {\includegraphics[height=4.6cm]{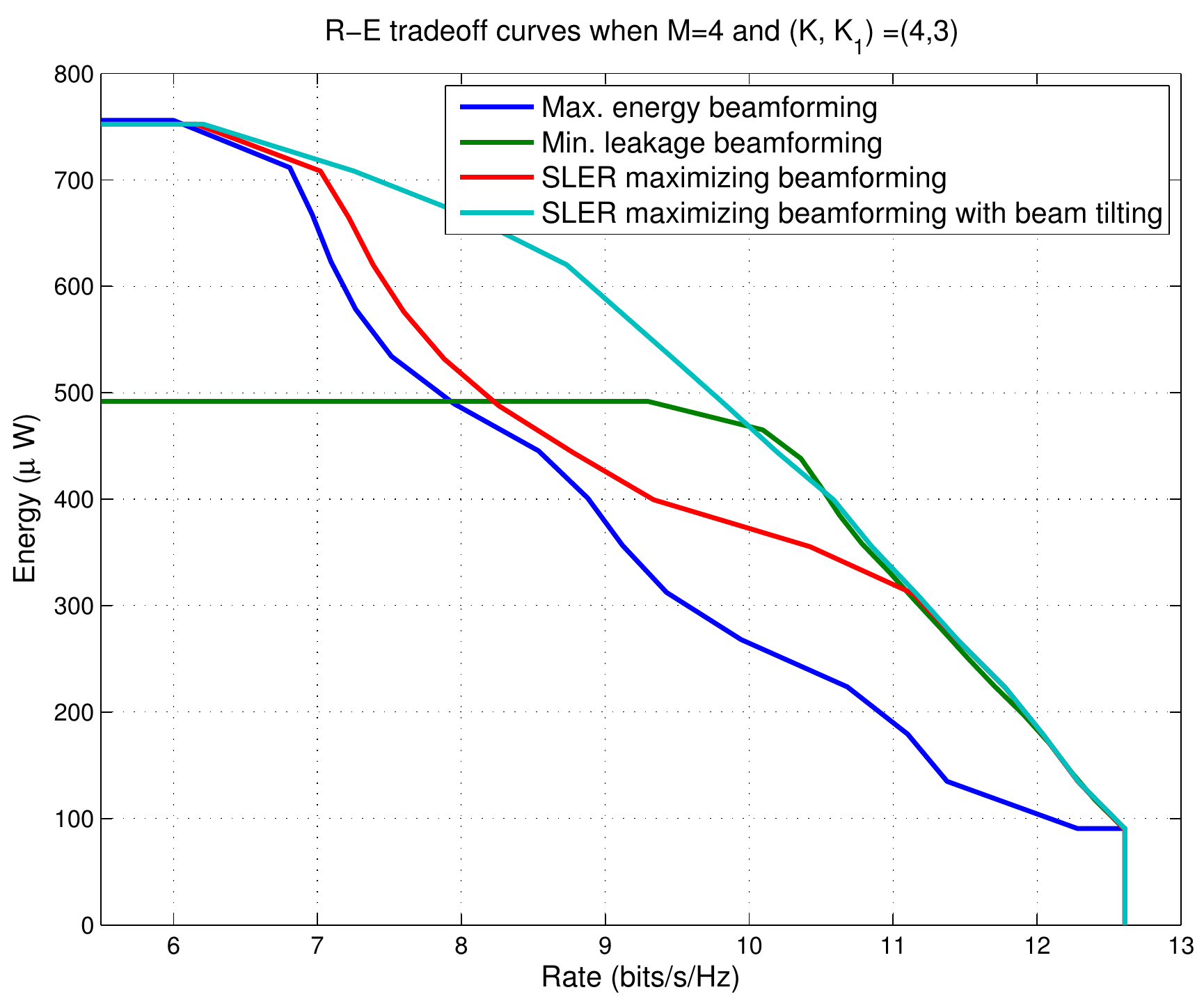}}
 \subfigure[]
  {\includegraphics[height=4.6cm]{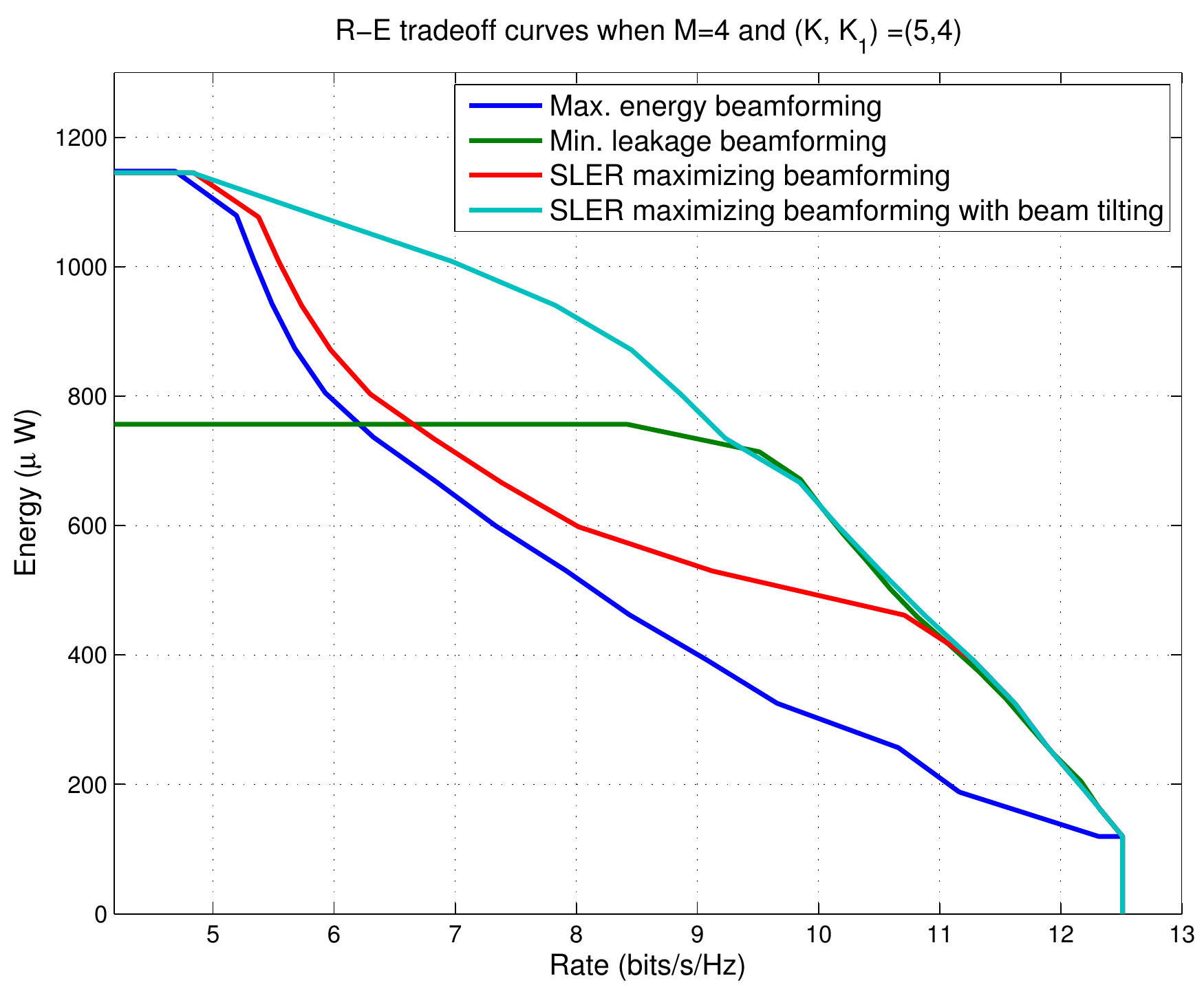}}
 \caption{R-E tradeoff curves for
SLER maximizing beamforming with/without beam tilting when (a) $(K,K_1) =(4,3)$ (b) $(K,K_1) =(5,4)$.} \label{R_E_tradeoff_GSVD_tilt}
\end{figure}
\end{remark}

\section{Conclusion}
\label{sec:conc} In this paper, we have investigated the joint
wireless information and energy transfer in K-user MIMO IFC.
The exact R-E curve for general K-user MIMO IFC is not known, but we have shown that the optimal energy transmitter's strategies for three different scenarios - i) multiple EH receivers and a single ID receiver, ii) multiple IDs and a single
EH, and iii) multiple IDs and multiple EHs - also become optimal for the {\it{properly-transformed}} two-user MIMO IFC. Accordingly,
we have found a common necessary condition of the optimal transmission strategy, in which all the transmitters transferring energy exploit a
rank-one energy beamforming. Furthermore, given the rank-one beamforming at the energy transmitters, we have
also developed the iterative algorithm for the non-convex optimization problem of the achievable rate-energy region. By comparing three different rank-one beamforming - MEB, MLB, and SLER maximizing beamforming, we can find that MEB and MLB either maximize the harvested energy or the information rate, but SLER maximizing beamforming scheme pursues both in a well-balanced way, showing a wider R-E region than that achieved by both MEB and MLB. Interestingly, when the number of information transmitters increases (interference-limited information transfer system where the interference due the information transmitters is dominant), the optimal strategy at the energy transmitters becomes close to MEB method. In contrast, when the number of energy transmitters increases, the beam steering as well as the power reduction affects the information rate performance, which lead us to develop the SLER maximizing beamforming with beam tilting. If the transmitters know the global CSI and a centralized optimization is possible, previous results motivate the energy beam alignment in cooperation with other transmitters. Finally, we have proposed an efficient SLER-based EH transceiver selection method which improves the achievable rate-energy region further.

Motivated by the information transfer \cite{QShi}, our approach can be extended to the MIMO interference broadcast channel (IBC). In addition, if we consider the power splitting method, new variables for the power splitting ratio at the receivers should be optimized in conjunction with the transmission strategy, which will be another challenging future work.

\renewcommand\baselinestretch{0.7}
\bibliographystyle{IEEEtran}
\bibliography{IEEEabrv,myref}

\begin{thebibliography}{10}
\providecommand{\url}[1]{#1}
\csname url@rmstyle\endcsname
\providecommand{\newblock}{\relax}
\providecommand{\bibinfo}[2]{#2}
\providecommand\BIBentrySTDinterwordspacing{\spaceskip=0pt\relax}
\providecommand\BIBentryALTinterwordstretchfactor{4}
\providecommand\BIBentryALTinterwordspacing{\spaceskip=\fontdimen2\font plus
\BIBentryALTinterwordstretchfactor\fontdimen3\font minus
  \fontdimen4\font\relax}
\providecommand\BIBforeignlanguage[2]{{%
\expandafter\ifx\csname l@#1\endcsname\relax
\typeout{** WARNING: IEEEtran.bst: No hyphenation pattern has been}%
\typeout{** loaded for the language `#1'. Using the pattern for}%
\typeout{** the default language instead.}%
\else
\language=\csname l@#1\endcsname
\fi
#2}}

\bibitem{3GPPMTC}
\emph{Study on Enhancements for {MTC}}, 3GPP TR Std. TR 22.888, v.0.4.0, 2011.

\bibitem{Soljacic}
A.~Kurs, A.~Karalis, R.~Moffatt, J.~D. Joannopoulos, P.~Fisher, and
  M.~Soljacic, ``Wireless power transfer via strongly coupled magnetic
  resonances,'' \emph{Science}, vol. 137, no.~83, pp. 83--86, Sept. 2007.

\bibitem{Yates}
M.~Pi{\~{n}}uela, P.~Mitcheson, and S.~Lucyszyn, ``Ambient {RF} energy
  harvesting in urban and semi-urban environments,'' \emph{{IEEE} Trans.
  Microwave Theory Tech.}, vol.~61, no.~7, pp. 2715--2726, July 2013.

\bibitem{Vullers}
R.~J.~M. Vullers, R.~V. Schaijk, I.~Doms, C.~V. Hoof, and R.~Merterns,
  ``Micropower energy harvesting,'' \emph{Solid-State Electronics}, vol.~53,
  no.~7, pp. 684--693, July 2009.

\bibitem{Fiez}
T.~Le, K.~Mayaram, and T.~Fiez, ``Efficient far-field radio frequency energy
  harvesting for passively powered sensor networks,'' \emph{{IEEE} J.
  Solid-State Circuits}, vol.~43, no.~5, pp. 1287--1302, May 2008.

\bibitem{Zhang1}
R.~Zhang and C.~K. Ho, ``{MIMO} broadcasting for simultaneous wireless
  information and power transfer,'' \emph{{IEEE} Trans. Wireless Commun.},
  vol.~12, no.~5, pp. 1989--2001, May 2013.

\bibitem{Zhang2}
L.~Liu, R.~Zhang, and K.~Chua, ``Wireless information transfer with
  opportunistic energy harvesting,'' \emph{{IEEE} Trans. Wireless Commun.},
  vol.~12, no.~1, pp. 288--300, Jan. 2013.

\bibitem{KHuang1}
K.~Huang and E.~G. Larsson, ``Simultaneous information-and-power transfer for
  broadband wireless systems,'' \emph{IEEE Transactions on Singla Processing,
  to be published}, 2013.

\bibitem{Ozel}
O.~Ozel, K.~Tutuncuoglu, J.~Yang, S.~Ulukus, and A.~Yener, ``Transmission with
  energy harvesting nodes in fading wireless channels: Optimal policies,''
  \emph{{IEEE} J. Select. Areas Commun.}, vol.~29, no.~8, pp. 1732--1743, Sept.
  2011.

\bibitem{RRajesh}
R.~Rajesh, V.~Sharma, and P.~Viswanath, ``Information capacity of energy
  harvesting sensor nodes,'' in \emph{Proc. {IEEE} International Symposium on
  Information Theory, 2011}, July 2011, pp. 2363--2367.

\bibitem{RRajesh1}
------, ``Information capacity of an energy harvesting sensor node,''
  \emph{submitted to IEEE Transactions on Information Theory,
  http://arxiv.org/abs/1212.3177}, 2012.

\bibitem{KIshibashi}
K.~Ishibashi, H.~Ochiai, and V.~Tarokh, ``Energy harvesting cooperative
  communications,'' in \emph{Proc. {IEEE} International Symposium on Personal,
  Indoor and Mobile Radio Communications, 2012}, Sept. 2012, pp. 1819--1823.

\bibitem{ANasir}
A.~A. Nasir, X.~Zhou, S.~Durrani, and R.~A. Kennedy, ``Relaying protocols for
  wireless energy harvesting and information processing,'' \emph{{IEEE} Trans.
  Wireless Commun.}, vol.~12, no.~7, pp. 3622--3636, July 2013.

\bibitem{YLuo}
Y.~Luo, J.~Zhang, and K.~B. Letaief, ``Optimal scheduling and power allocation
  for two-hop energy harvesting communication systems,'' \emph{{IEEE} Trans.
  Wireless Commun.}, vol.~12, no.~9, pp. 4729--4741, Sept. 2013.

\bibitem{XuZhang}
J.~Xu, L.~Liu, and R.~Zhang, ``Multiuser {MISO} beamforming for simultaneous
  wireless information and power transfer,'' \emph{submitted to IEEE
  Transactions on Signal Processing, http://arxiv.org/abs/1303.1911}, 2013.

\bibitem{ZhouZhangHo2}
X.~Zhou, R.~Zhang, and C.~K. Ho, ``Wireless information and power transfer in
  multiuser {OFDM} systems,'' \emph{submitted to IEEE Transactions on Wireless
  Communications, http://arxiv.org/abs/1308.2462}, 2013.

\bibitem{KwanSchober}
D.~W.~K. Ng, E.~S. Lo, and R.~Schober, ``Energy-efficient resource allocation
  in multiuser {OFDM} systems with wireless information and power transfer,''
  in \emph{Proc. {IEEE} Wireless Communications and Networking Conference,
  2013}, Apr. 2013, pp. 3823--3828.

\bibitem{Tutuncuoglu1}
K.~Tutuncuoglu and A.~Yener, ``Sum-rate optimal power policies for energy
  harvesting transmitters in an interference channel,'' \emph{Journal of
  Communications and Networks}, vol.~14, no.~2, pp. 151--161, Apr. 2012.

\bibitem{Tutuncuoglu2}
------, ``Transmission policies for asymmetric interference channels with
  energy harvesting nodes,'' in \emph{Proc. {IEEE} International Workshop on
  Computational Advances in Multi-sensor Adaptive Processing, 2011}, Dec. 2011,
  pp. 197--200.

\bibitem{KHuang2}
K.~Huang and V.~K.~N. Lau, ``Enabling wireless power transfer in cellular
  networks: architecture, modeling and deployment,'' \emph{submitted to IEEE
  Transactions on Signal Processing, http://arxiv.org/abs/1207.5640}, 2012.

\bibitem{ChenLiChang}
C.~Shen, W.~Li, and T.~Chang, ``Wireless information and energy transfer in
  multi-antenna interference channel,'' \emph{submitted to IEEE Transactions on
  Signal Processing, http://xxx.tau.ac.il/abs/1308.2838}, 2013.

\bibitem{ParkBruno}
J.~Park and B.~Clerckx, ``Joint wireless information and energy transfer in a
  two-user mimo interference channel,'' \emph{{IEEE} Trans. Wireless Commun.},
  vol.~12, no.~8, pp. 4210--4221, Aug. 2013.

\bibitem{ZhouZhangHo}
X.~Zhou, R.~Zhang, and C.~K. Ho, ``Wireless information and power transfer:
  architecture design and rate-energy tradeoff,'' \emph{submitted to IEEE
  Transactions on Communications, http://arxiv.org/abs/1205.0618}, 2012.

\bibitem{Loyka}
S.~L. Loyka, ``Channel capacity of {MIMO} architecture using the exponential
  correlation matrix,'' \emph{{IEEE} Commun. Lett.}, vol.~5, no.~9, pp.
  369--371, Sept. 2001.

\bibitem{Scutari}
G.~Scutari, D.~P. Palomar, and S.~Barbarossa, ``The {MIMO} iterative
  waterfilling algorithm,'' \emph{{IEEE} Trans. Signal Processing}, vol.~57,
  no.~5, pp. 1917--1935, May 2009.

\bibitem{Shen}
C.~Shen, W.~Li, and T.~Chang, ``Simultaneous information and energy transfer: A
  two-user {MISO} interference channel case,'' in \emph{Proc. {IEEE} GLOBECOM,
  2012}, Dec. 2012.

\bibitem{RHorn}
R.~A. Horn, N.~H. Rhee, and W.~So, ``Eigenvalue inequalities and equalities,''
  \emph{Linear Algebra and Its Applications}, vol. 270, no. 1-3, pp. 29--44,
  Feb. 1998.

\bibitem{QShi}
Q.~Shi, M.~Razaviyayn, Z.~Luo, and C.~He, ``An iteratively weighted {MMSE}
  approach to distributed sum-utility maximization for a {MIMO} interfering
  broadcast channel,'' \emph{{IEEE} Trans. Signal Processing}, vol.~59, no.~9,
  pp. 4331--4340, Sept. 2011.

\bibitem{ParkClerckxJSAC}
J.~Park and B.~Clerckx, ``Joint wireless information and energy transfer with
  reduced feedback in {MIMO} interference channels,'' \emph{submitted to IEEE
  Journal on Selected Areas in Communication}, 2014.

\bibitem{Ipsen}
I.~C.~F. Ipsen and D.~J. Lee, ``Determinant approximations,'' \emph{submitted
  to Numer. Linear Algebra Appl., http://arxiv.org/abs/1105.0437}, 2011.

\bibitem{RHorn2}
R.~Horn and C.~Johnson, \emph{Topics in Matrix Analysis}.\hskip 1em plus 0.5em
  minus 0.4em\relax New York: Cambridge University Press, 1991.

\bibitem{JPark}
J.~Park, J.~Chun, and H.~Park, ``Generalised singular value decomposition-based
  algorithm for multi-user {MIMO} linear precoding and antenna selection,''
  \emph{IET Commun.}, vol.~4, no.~16, pp. 1899--1907, Nov. 2010.

\bibitem{JPark2}
J.~Park, J.~Chun, and B.~Jeong, ``Efficient multi-user {MIMO} precoding based
  on {GSVD} and vector perturbation,'' \emph{Signal Processing}, vol.~92,
  no.~2, pp. 611--615, Feb. 2012.

\bibitem{Magnus}
J.~R. Magnus and H.~Neudecker, \emph{Matrix Differential Calculus with
  Applications in Statistics and Economics}, 3rd~ed.\hskip 1em plus 0.5em minus
  0.4em\relax Chichester: John Wiley \& Sons, 2007.

\bibitem{DHarville}
D.~A. Harville, \emph{Matrix algebra from a statistician's perspective}.\hskip
  1em plus 0.5em minus 0.4em\relax Berlin: Springer, 2008.

\bibitem{SBoyd}
S.~Boyd and L.~Vandenberghe, \emph{Convex Optimization}, 7th~ed.\hskip 1em plus
  0.5em minus 0.4em\relax New York: Cambridge University Press, 2009.

\bibitem{XZhao}
X.~Zhao, P.~B. Luh, and J.~Wang, ``Surrogate gradient algorithm for
  {L}agrangian relaxation,'' \emph{Journal of Optimization Theory and
  Applications}, vol. 100, no.~3, pp. 699--712, Mar. 1999.

\bibitem{Bazaraa_book}
M.~S. Bazaraa, H.~D. Sherali, and C.~M. Shetty, \emph{Nonlinear Programming:
  Theory and Algorithms}, 3rd~ed.\hskip 1em plus 0.5em minus 0.4em\relax New
  York: John Wiley and Sons, 1993.

\bibitem{ZhangLiangCui1}
R.~Zhang, Y.~Liang, and S.~Cui, ``Dynamic resource allocation in cognitive
  radio networks,'' \emph{{IEEE} Signal Processing Mag.}, vol.~27, no.~3, pp.
  102--114, May 2010.

\bibitem{NegroShenoy}
F.~Negro, S.~P. Shenoy, I.~Ghauri, and D.~T. Slock, ``Weighted sum rate
  maximization in the {MIMO} interference channel,'' in \emph{Proc. {IEEE}
  International Symposium on Personal, Indoor and Mobile Radio Communications,
  2010}, Sept. 2010, pp. 684--689.

\bibitem{YeBlum}
S.~Ye and R.~S. Blum, ``Optimized signaling for {MIMO} interference systems
  with feedback,'' \emph{{IEEE} Trans. Signal Processing}, vol.~51, no.~11, pp.
  2839--2848, Nov. 2003.

\bibitem{AlShatri}
H.~Al-Shatri and T.~Weber, ``Optimizing power allocation in interference
  channels using {D.C.} programming,'' in \emph{Proc. the 8th International
  Symposium on Modeling and Optimization in Mobile, Ad Hoc and Wireless
  Networks}, May 2010, pp. 360--366.

\bibitem{XuLeNgoc}
Y.~Xu, T.~Le-Ngoc, and S.~Panigrahi, ``Global concave minimization for optimal
  spectrum balancing in multi-user {DSL} networks,'' \emph{{IEEE} Trans. Signal
  Processing}, vol.~56, no.~7, pp. 2875--2884, July 2008.

\end{thebibliography}

\end{document}